\documentclass[a4paper, twocolumn, accepted=2021-08-20]{quantumarticle}
\pdfoutput=1

\usepackage[utf8]{inputenc}
\usepackage[T1]{fontenc}

\usepackage{graphicx}
\usepackage{hyperref}
\usepackage[normalem]{ulem}
\usepackage{bm}
\usepackage{mathrsfs}
\usepackage{authblk}
\usepackage{enumerate}
\usepackage[numbers, sort&compress]{natbib}
\usepackage{float}

\usepackage{amsmath}
\usepackage{amssymb}
\usepackage{amsthm}
\usepackage{amsfonts}
\usepackage{mathtools}
\usepackage{bigints}
\usepackage{xcolor}

\usepackage{braket}

\newtheorem{theorem}{Theorem}

\newtheorem{corollary}{Corollary}
\newtheorem{lemma}{Lemma}

\theoremstyle{definition}

\numberwithin{equation}{section}

\begin{document}

\title{Kelly Betting with Quantum Payoff: a continuous variable approach}
\author[1]{Salvatore Tirone}
\author[1]{Maddalena Ghio}
\author[1]{Giulia Livieri}
\author[2]{Vittorio Giovannetti}
\author[1]{Stefano Marmi}
\affil[1]{Scuola Normale Superiore, I-56126 Pisa, Italy}
\affil[2]{NEST, Scuola Normale Superiore and Istituto Nanoscienze-CNR, I-56126 Pisa, Italy}
\date{}
\maketitle
\begin{abstract}
The main purpose of this study is to introduce a semi-classical model describing betting scenarios in which, at variance with conventional 
approaches, the payoff of the gambler is encoded into the internal degrees of freedom of a quantum memory element. In our scheme, we assume that the invested capital is explicitly associated with the quantum analog of the free-energy (i.e. ergotropy functional by
Allahverdyan, Balian, and Nieuwenhuizen) of a single mode of the electromagnetic radiation which, depending on the outcome of the betting, experiences attenuation or amplification processes which model losses and winning events. The resulting stochastic evolution of the quantum memory  resembles
the dynamics of random lasing which we characterize within the theoretical setting of Bosonic Gaussian channels. As in the classical Kelly Criterion for optimal betting, we define the  asymptotic doubling rate of the model and identify the optimal gambling strategy for fixed odds and probabilities of winning. The performance of the model are hence studied as a function of the input capital state under the assumption that the latter belongs to the set
of Gaussian density matrices (i.e. displaced, squeezed thermal Gibbs states) revealing that the best option for the gambler is
to devote all their initial resources into coherent state amplitude. 
\end{abstract}

\section{Introduction}\label{sec:Intro}
\noindent The problem of assessing the maximum growth of an optimal investment, or equivalently of the maximization of the long-run interest rate, is one of the central issues in quantitative finance. An especially simple but fundamental example is provided by horse-race markets, i.e. markets with the property that at every time step one of the assets pays off and all the other assets pay nothing (the wealth invested is lost completely). Remarkably, horse-race markets are very special cases of general markets, being, in a sense, the extremal points of the distribution of asset returns \citep{bell1988game, iyengar2000growth}. As such they have been extensively investigated:  they possess many of the usual attributes of financial markets but they also have the important additional property that each bet has a well defined end point at which its value becomes certain. This is rarely the case in finance, where asset values depend on an intrinsically uncertain future \cite{williams2005information}. In the 1950s, J. L. Kelly, Jr. worked out a striking interpretation of the rate of transmission of information over a noisy communication channel in the setting of gambling theory~\cite{kelly}. 
In this context, the input symbols to the communication line were interpreted as the outcomes of an uncertain event on which betting is possible.
Exploiting this construction, under ``fair" odds and  independent and identically distributed (\textit{i.i.d.}) horse-race markets assumptions (with no transactions costs), Kelly proved the optimality of 
proportional betting strategies; specifically, he proved that the asymptotic growth rate of the cumulative wealth is maximal when the fraction of capital bet on each horse is proportional to its true winning probability. Remarkably, over the years, both theory and practice of the Kelly criterion in gambling and investment have developed prolifically \cite{maclean2011kelly}.\\
\indent In this paper, we present a theory of classical betting with a quantum payoff which investigates the behaviour of an economic agent whose  initial capital is represented  by the value of a quantum resource encoded into the state of a quantum memory that faces uncertainty under action of a stochastic environment.
In order to achieve this, we rely on the versatility of  Bosonic Gaussian channels (BGCs) to mimic the action of the random gains and losses that affect the capital  in a horse-race market.\\
\indent BGCs play a fundamental role in quantum information theory
\cite{HOLEVOBOOK, WILDEBOOK} 
where they act as the proper counterparts of Gaussian channels of classical information theory~
(see, for instance \cite{holevowerner2001, caruso2006one,cerf2014}and references therein). Formally speaking, they can be identified
with the set of quantum evolutions which preserve the ``Gaussian character" of the transmitted signals and possess the striking feature, heavily used in this paper, to  be closed under composition \cite{serafini2017qcv}. At the physical level  BGCs describe the most relevant  transformations an optical pulse  may experience  when 
propagating in a  noisy environment. In particular, they provide a proper
representation of  attenuation and  amplification processes for continuous variable quantum systems \cite{holevowerner2001, cerf2014}.  Differently from their classical counterparts, which can be always realized without extra added noise effects,  
BGCs include vacuum fluctuation and display a non-commutative character
which has an intrinsic quantum mechanical origin. As a result, the setting we introduce here is inherently different from previously analyzed horse-market models \cite{williams2005information, hausch2008efficiency}, giving rise to a non-trivial generalization of classical betting procedures.\\
\indent The connection between game theory and quantum information has a long history dating back to the pioneering results presented in
Refs.~\cite{meyer1999quantum,goldenberg1999quantum,eisert1999quantum}.
The common denominator of these works is to grant players access to quantum resources
like entanglement or mere quantum coherence, 
that help them in defining new strategies, thus enlarging the range of  possible operations.
Our approach on the contrary relies on quantizing not the protocol itself but only the resource that is used to encode both the invested capital and the payoff of the game. Specifically,
in our model we identify the  capital  of our economic agent with the  quantum ergotropy functional ${\cal E}(\hat{\rho})$ \cite{alicki2013qb} associated to the density matrix  $\hat{\rho}$ of a single Bosonic mode~$A$ that is under the control of the 
agent. Ergotropy  plays an important role in many different contexts, spanning from the characterization of optimal thermodynamical cycles~\cite{REV1,REV2}
to energetic instabilities~\cite{PUSZ,LENARD}, and to the study of quantum batteries models \cite{alicki2013qb, andolina2019, farina2019}: it
represents the maximum mean energy decrement obtainable when forcing the quantum system $A$ to undergo unitary evolutions induced by  cyclic  external modulations of its Hamiltonian. The use of the ergotropy as a  quantifier of the wealth  of the agent is legitimated by the fact that  in the Kelvin-Planck formulation of the second law of thermodynamics, ${\cal E}(\hat{\rho})$ 
can be interpreted as the maximum amount of useful energy (work) which we can extract form the state $\hat{\rho}$ 
via reversible coherent (i.e. unitary) transformations, i.e, in the absence of extra resources and without collateral dissipation events~\cite{Niedenzu2019conceptsofworkin}.
This choice  follows also the long history of analogy making between economics (and finance) and (stastical) 
thermodynamics discussed for instance in \cite{smith2008classical, kim2015distribution, saslowtherm}, and references therein. It is worth observing also that our representation of the invested capital reminds somehow the proposal of Orrell which used the energy (not the ergotropy) of an harmonic oscillator to describe the price of  transactions occurring during a trading cycle \cite{ORRELL2020122928, Orell2}.

Regarding the uncertainty of the outcome of the  horse race,  we describe it with the random selection of one among a collection of  a finite set of one-mode BGCs  operating on~$A$. Their action  includes both an attenuation part -- associated with the partition of the capital determined by the strategy adopted by the gambler -- and an amplification part --associated instead with the odds offered by the bookmaker. Each channel is therefore randomly selected and applied to $A$ according to  the winning probability $p_j$ of the corresponding horse. The whole procedure is iterated to reproduce the \textit{i.i.d.} horse-race markets framework, resulting into a stochastic
trajectory for $A$ in which the state of the system at the discrete time step $t$, is obtained by a concatenation of $t$ elements  of the selected BGCs set.
We notice that  from a practical standpoint our construction is a simple idealization of a general random lasing process~\cite{Wiersma2008}, so the formalism we introduce here could also serve as a modelization
of the thermodynamic aspects of light amplification in disordered materials.
The resulting payoff for the gambler is therefore computed by evaluating the ergotropy functional
on the final state of the problem. The aim of our work is characterizing the  probability distribution of the resulting accumulated wealth and determining which, among all possible choices of the input state $\hat{\rho}_{0}$ of $A$, provides the best performance in terms of capital-gain for a given game strategy. It is worth stressing that  at the level of the mere internal energy of the system $A$, irrespective of the choice of the initial state $\hat{\rho}_{0}$, the process~(\ref{rho0torhot}) behaves essentially as its classical horse-market counterpart, a part from a noisy contribution, where the BGCs are replaced by standard betting operations. On the contrary, at the level of ergotropy, the non-commutative nature of the involved quantum operation, combined with the classical randomness of the bet, results in the injection of extra noise terms which  affect negatively the performance of the procedure, paving the way for a nontrivial optimization of the initial resource $\hat{\rho}_{0}$. In our analysis we shall address this problem by explicitly restricting the class of allowed inputs to the special (yet not trivial) subset of Gaussian density matrices \cite{serafini2017qcv}. This choice  on one hand allows for some mathematical  simplification while on the other hand 
is physically motivated by the fact that this special collection of  states includes a rather large class of  electromagnetic configurations (notably thermal, coherent, and squeezed signals) 
that are  very well under experimental control.\\
\indent The rest of the paper is organized as follows: we start in Sec.~\ref{sec:Preliminaries} with a brief review on BGCs.  The Kelly model with quantum payoff is hence introduced in Sec.~\ref{sec:PlayingField}. In Sec.~\ref{sec:KellyQuantumClassical} we study its statistical properties and present 
some numerical analysis. Conclusions are presented in Sec.~\ref{sec:Conclusions}.
The manuscript also contains an Appendix devoted to elucidate some technical aspects of the problem.

\section{Preliminaries}\label{sec:Preliminaries}

\subsection{Gaussian states}\label{sec:BosonicGaussianStates}

Let $\mathcal{H}$ be a complex separable Hilbert space  of infinite dimension associated with one mode  $A$ of the electromagnetic radiation (or equivalently a 1-dimensional, quantum harmonic oscillator) of frequency $\nu$. The model is fully described by the assignment of the canonical coordinate vector   $\hat{r} := (\hat{q},\hat{p})^T$ and  the quadratic Hamiltonian 
\begin{equation}
\hat{H} := h\nu\;  \frac{\hat{r}^T \cdot \hat{r}}{2} = h\nu \left( \frac{\hat{p}^2 + \hat{q}^2}{2}
\right)\;,
\label{eq:electroHamiltonian}
\end{equation}
expressed in terms of (rescaled) position and momentum operators $\hat{q}$ and $\hat{p}$ which satisfy the (Bosonic) canonical commutation relation: $[\hat{q},\hat{p}]=i$ ($h$ being the Planck constant).
We let $\mathcal{D}(\mathcal{H})$ be the set of the positive semidefinite operators with trace $1$, and we call $\hat{\rho}$ a quantum state of $A$ if $\hat{\rho} \in \mathcal{D}(\mathcal{H})$. For each element of such set we can then associate a column vector of~$\mathbb{R}^2$, 
\begin{equation}\label{DefM} 
m := \text{Tr}[\hat{\rho}\; \hat{r}]\;,
\end{equation}
given by the expectation values of the canonical coordinates, and 
a $2\times 2$ covariance matrix $\sigma$ whose elements are
\begin{equation}\label{sigma} 
\sigma_{ij} := \text{Tr}[\hat{\rho}\{(\hat{r}_i-m_i),(\hat{r}_j-m_j)\}],
\end{equation}
with $\{\,,\}$ denoting the anti-commutator operation, which describes instead the second order moments of the coordinate distributions on $\hat{\rho}$. Although no constraints hold on $m$, on $\sigma$ quantum uncertainty relations impose  the inequality $\sigma \geq\sigma_2$, with $\sigma_2 := \tiny{
\left[ \begin{array}{cc} 0 & -i \\ i &0\end{array}\right]}$ being the second Pauli matrix. 
An exhaustive description of $\mathcal{D}(\mathcal{H})$ is finally provided by the
\textit{characteristic function} formalism which, given a density matrix $\hat{\rho}$ allows us to faithfully represents it in terms of the following complex functional 
\begin{equation}\label{eq:charfunc}
\phi(\hat{\rho};z) := \text{Tr}\left[\hat{\rho} \hat{V}(z)\right]\;, 
\end{equation}
where  $z:=(x,y)^T$ is a  column vector of $\mathbb{R}^2$, and
$\hat{V}(z):= \exp  [i \hat{r}^T\cdot z ]$ is the Weyl operator of the system. 
In this context, the density matrix $\hat{\rho}\in  \mathcal{D}(\mathcal{H})$ is thus said to be Gaussian if its associated characteristic function has a  Gaussian form, i.e. namely if we have 
\begin{equation} \label{GAUSSCHAR} 
\phi_G(\hat{\rho};z) = \exp[{-\tfrac{1}{4}z^T\sigma z + im^T\cdot z}]\;, 
\end{equation}
with $m$ and $\sigma$ defined as in Eqs.~(\ref{DefM}) and (\ref{sigma}) respectively~\cite{serafini2017qcv}. 
At the physical level Gaussian states represent thermal Gibbs density matrices of generic Hamiltonians which are (non-negative) quadratic forms of the canonical vector $\hat{r}$. A convenient parametrization is given by the following expression 
\begin{equation} \label{PARGAUS} 
\hat{\rho} = \hat{V}(m)\hat{S}({\zeta})\hat{\rho}_{\beta}\hat{S}^{\dagger}(\zeta)\hat{V}^{\dagger}(m),
\end{equation}
where for $\beta\geq 0$, $\hat{\rho}_{\beta} := \exp[ - \beta \hat{H}] /Z_\beta$  defines a thermal state of the mode Hamiltonian~(\ref{eq:electroHamiltonian}) with inverse temperature $\beta$ ($Z_\beta=\mbox{Tr}[ \exp[ - \beta \hat{H}]]$ being the associated partition function) and where
for $\zeta:= |\zeta| e^{i \varphi}$ complex we introduce the squeezing operator
\begin{eqnarray}
\hat{S}({\zeta}):= \exp\left[i \tfrac{|\zeta|}{2} \hat{r}^T R(\varphi) \hat{r}\right]\;,
\end{eqnarray} 
with $R(\varphi):= \cos(\varphi) \sigma_1 - \sin(\varphi)  \sigma_3$
being a rotation matrix defined in terms of the
first and third Pauli matrices $\sigma_1:= \tiny{
\left[ \begin{array}{cc} 0 & 1 \\ 1 &0\end{array}
\right]}$ and $\sigma_3:= \tiny{
\left[ \begin{array}{cc} 1 & 0 \\ 0 &-1\end{array}
\right]}$. 
With this choice we get that the associated characteristic function of~(\ref{PARGAUS}) is
in the form (\ref{GAUSSCHAR}) with 
\begin{eqnarray} 
\sigma &=&  (2 {{{n}}} +1) [\cosh(2|\zeta|)  \mathbb{I}\nonumber \\
&&+ \sinh(2|\zeta|)(\sin(\varphi) \sigma_1 - \cos(\varphi) \sigma_3)]\;, 
\end{eqnarray} 
where  $\mathbb{I}$  is the $2\times 2$ identity 
matrix and ${{{n}}} := 1/(e^{\beta h \nu} -1)$ is the average photon number of $\hat{\rho}_{\beta}$. Special cases of Gaussian states are the coherent states which often are referred to as semi-classical configurations of the electro-magnetic field obtained
by setting ${{{n}}}=\zeta=0$ in (\ref{PARGAUS}). 

\subsection{One-mode Bosonic Gaussian channels}\label{sec:BosonicGaussianChannels}

A BGC is a Completely Positive, Trace Preserving (CPTP) dynamical  transformation $\Phi$
\cite{HOLEVOBOOK,WILDEBOOK} defined by assigning a column  vector $\lambda$ of $\mathbb{R}^2$, and two  positive  real matrices $X$ and $Y$ fulfilling the constraint 
\begin{eqnarray} 2 Y + X^T \sigma_2 X \geq \sigma_2 \;.  \label{CPTP} 
\end{eqnarray} 
When operating on  a generic $\hat{\rho} \in \mathcal{D}(\mathcal{H})$,  ${\Phi}$ sends it into a new density matrix ${\Phi}(\hat{\rho})$ whose characteristic function is expressed by the formula
\begin{equation}\label{eq:carachteristic}
\phi({\Phi}(\hat{\rho});z) = \phi(\hat{\rho};Xz)\exp\left[-\tfrac{1}{2} z^T Y z  + i\lambda^T \cdot z\right]\;.
\end{equation}
At the level of first and second moments, this mapping  induces the transformation 
\begin{eqnarray}
m&\rightarrow& m' = X^T m + \lambda\;, \nonumber  \\ 
\sigma &\rightarrow&  \sigma'= X^T \sigma X + 2 Y\;, \label{EVOLUZIONESIGMA} 
\end{eqnarray} 
the couple $m$, $\sigma$ being associated with the input state $\hat{\rho}$, and  
$m'$, $\sigma'$ being instead associated with the output state ${\Phi}(\hat{\rho})$.
These transformations inherit their name from the notable fact that 
when applied to  an arbitrary Gaussian input they will produce Gaussian 
outputs. Most importantly for us, these processes are closed under super-operator composition, meaning that
given two {BGC} maps ${\Phi_1}$ and ${\Phi_2}$ identified respectively by the triples $\lambda_1,X_1, Y_1$ and 
$\lambda_2,X_2, Y_2$,
the transformation  ${\Phi_2} \circ
{\Phi_1}$ obtained by concatenating them in sequence
is still a {BGC} channel $\Phi_{3}$  characterized by the new triple
\begin{eqnarray}
\mu_3&:=& X_2^T \lambda_1 + \lambda_2 \nonumber \;,\\
X_3&:=& X_1 X_2 \nonumber \;,\\ 
Y_3&:=& X_2^T Y_1 X_2 + Y_2 
\label{COMPO}  \;.
\end{eqnarray} 
As already mentioned in the introduction, at physical level BGCs describe 
a vast variety of processes. 
In what follows we shall focus on the special subset of BGCs associated with
amplification and attenuation events, possibly accompanied by thermalization effects.
These maps have  $\lambda=0$ and   $X$, $Y$  proportional to $\mathbb{I}$. In view of the constraint~(\ref{CPTP}) and following a rather standard convention, we find it useful to  express them in terms of two non-negative constants $g, N$  as 
$X = g\; \mathbb{I}$, $Y =\alpha\; \mathbb{I}$ with 
\begin{eqnarray}
\alpha := \left\{ \begin{array}{ll} 
|g^2 - 1| \left(N + \frac{1}{2}\right)
& \mbox{for $g\neq 1$} \;, \\
N  & \mbox{for $g=1$}\;, \end{array} \right. \label{defALPHA} 
\end{eqnarray}
and use the notation $\Phi[g;{\alpha}]$ to represent the associated channel. 
Notice that in this case  Eq.~(\ref{EVOLUZIONESIGMA}) gets simplified as
\begin{eqnarray} \label{defm} 
m&\rightarrow& m' = m\,g\;,\\ \label{defsigma} 
\sigma &\rightarrow&  \sigma' = g^2\,\sigma + 2 \alpha \; \mathbb{I}\;,
\end{eqnarray}
which allows us to interpret $g$ as a stretching parameter for the first moments of the state,
and $\alpha$ as an additive noise term.
Furthermore  
from~(\ref{COMPO}) it trivially follows that 
the BGC subset formed by the maps $\Phi[g;\alpha]$  is also closed under channel composition leading to the following identity 
\begin{eqnarray} \label{COMPO1} 
\Phi[{g_2};{\alpha_2}]\circ \Phi[{g_1};{\alpha_1}] =
\Phi[{g_3};{\alpha_3}]
\end{eqnarray} 
with
\begin{eqnarray} 
{g_3:= g_2g_1} \;, \qquad 
{\alpha_3:= g_2^2 \alpha_1 + \alpha_2}\;. 
\end{eqnarray} 
For $g<1$,  $\Phi[g;{\alpha}]$  defines  a process where
the mode $A$ interacts dispersively with an external thermal Bosonic envinronment characterized by a mean photon number $N$ (thermal attenuator channel). 
In particular for $g= \eta \in [0,1[$ and $N=0$ (i.e. $\alpha = (1-\eta^2)/2$) we get \begin{eqnarray}\label{LOSSY} 
\Phi_\eta^{(\text{loss})} :=  \Phi[\eta;{(1-\eta^2)/2}]\;, 
\end{eqnarray} 
which describes a  pure loss of energy from the system, corresponding to the evolution the mode $A$ experiences when passing through a  medium (beam-splitter) of transmissivity $\eta=g$. A special case of particular interest for us is  depicted in Fig.~\ref{fig:BS} where a collection of beam-splitters (semitransparent mirrors~\cite{OPTICS})
are arranged to split an input optical signal $\hat{\rho}$ into a different direction by properly mixing it with the vacuum.  In this construction, while the joint state emerging from the set up may exhibit nontrivial correlations, the local density matrices associated 
to the various output ports are described as $\Phi_\eta^{(\text{loss})}[\hat{\rho}]
$ with $\eta$ being the corresponding effective transmissivity.

\begin{figure}[t!]
	\centering
	\includegraphics[width=\linewidth]{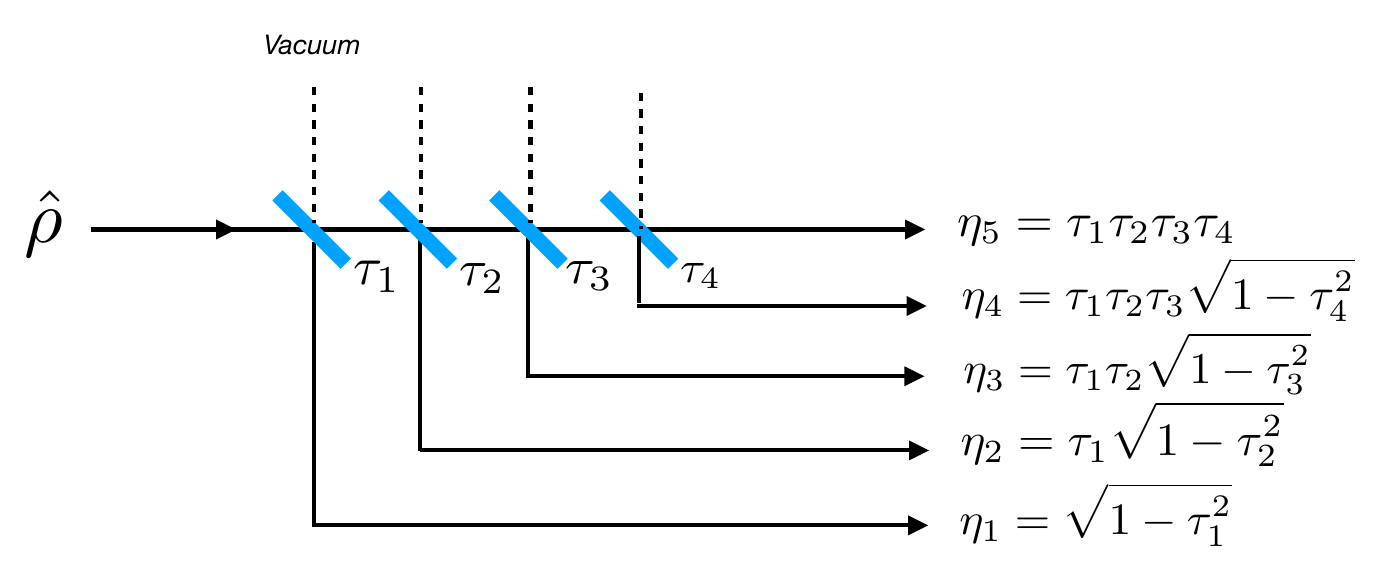}
	\caption{Schematic of an optical splitter formed by a collection of  beam-splitters (blue elements) 
	 that induce a fractioning of the input state
	into a collection of outputs which, locally, are associated with the action of purely lossy channels~(\ref{LOSSY}). In the scheme $\tau_j$ are the transmissivities of the individual beam-splitters, while
	$\eta_j$ are the resulting transmissivities that emerge from their combinations which, by construction,
	fulfil the normalization condition~$\sum_j \eta_j^2=1$. }
	 \label{fig:BS}
\end{figure}

For $g> 1$, $\Phi[g;{\alpha}]$ describes instead the  interaction of $A$ 
and the thermal bath in terms of an amplification event (thermal amplifier).
In this case for $g= k >1$ and $N=0$ (i.e. $\alpha = (k^2-1)/2$) the transformation
\begin{eqnarray}\label{AMPY} 
\Phi_k^{(\text{amp})} := \Phi[k;{(k^2-1)/2}] \; 
\end{eqnarray} 
defines a purely amplifying process induced by the  interaction between $A$ and
the vacuum mediated by  a non-linear optical crystal~\cite{OPTICS}.
Finally, we observe that in the special case where $g=1$, the transformation
\begin{eqnarray}\label{ADD} 
\Phi_1^{(\text{add})} := \Phi[1,{N}] \;, 
\end{eqnarray} 
corresponds to the {\it noise-additive} channel which describes a random Gaussian displacement in the phase space~\cite{holevowerner2001}. 

\subsection{Average energy and ergotropy}\label{sec:ExtractableWork}
The mean energy of a generic state $\hat{\rho}$ of the mode $A$ 
is given by
\begin{equation} \label{DEFENERG} 
E(\hat{\rho}) := \text{Tr}[\hat{\rho}\hat{H}] = \frac{\text{Tr}(\sigma)}{4} + \frac{m^2}{2}\;,
\end{equation}
where $m^2=m^T\cdot m$ is the square norm of $m$, where
in the second identity we made explicit use of Eqs.~(\ref{DefM}) and (\ref{sigma}), and
where without loss of generality we set  $h\nu=1$.
For a Gaussian state, exploiting the parametrization introduced in Eq.~(\ref{PARGAUS}), this can be cast
in a formula that highlights the 
squeezing, displacement, and thermal contributions to the final result, i.e. 
\begin{eqnarray} \label{DEFENERGGAUS} 
E(\hat{\rho}) =\frac{(2{{{n}}}+1)\cosh(2|\zeta|) + {m}^2}{2} \;.\end{eqnarray} 
Equations~(\ref{DEFENERG}) and  (\ref{DEFENERGGAUS}) represent the total energy stored in $\hat{\rho}$ which one could  extract  from the $A$,
e.g. by  putting it in thermal contact with  an external zero-temperature environment.
Yet, as pointed out in Ref.~\cite{allahverdyan2004maximal, ergobrown}, not all such energy will be nicely converted into useful work: part of 
$E(\hat{\rho})$ will indeed necessarily emerge as dissipated heat. The quantity that instead properly quantifies the maximum amount of extractable work is given by the ergotropy. At variance with $E(\hat{\rho})$, the latter is a  nonlinear functional of $\hat{\rho}$ that we can compute as 
\begin{equation} \label{ERGOTROPY} 
\mathcal{E}(\hat{\rho}) :=  E(\hat{\rho}) - \min_{\hat{U}}E(\hat{U}\,\hat{\rho}\,\hat{U}^{\dagger}),
\end{equation}
where the minimization is performed over all unitary transformations $\hat{U}$. For a one-mode Bosonic Gaussian state a closed expression in terms of $\sigma$ and $m$ is known which we shall
use in the following, i.e. \cite{farina2019}
\begin{eqnarray}
\mathcal{E}(\hat{\rho}) &=& \frac{\text{Tr}(\sigma)}{4} + \frac{{m}^2}{2} - \frac{\sqrt{\text{Det}(\sigma)}}{2} \nonumber\\
&=& \frac{(2{{{n}}}+1)(\cosh(2|\zeta|)-1) + {m}^2}{2}\;. \label{ERGOG} 
\end{eqnarray}
This equation shows that 
the gap between $E(\hat{\rho})$ and $\mathcal{E}(\hat{\rho})$
only depends upon the thermal component of $\hat{\rho}$, i.e.
\begin{eqnarray} \label{DIFFE} 
E(\hat{\rho})-\mathcal{E}(\hat{\rho})=  {{{n}}} +1/2\;,
\end{eqnarray}
 reaching its minimum value  for pure (squeezed  
coherent) states (i.e. ${{{n}}}=0$) where it matches the vacuum energy level $1/2$.
Notice also that states $\hat{\rho}$ which are characterized by the same ergotropy value  form a continuous 
2D manifold in the phase space  associated with the variables
$m$, $\zeta$, and ${{{n}}}$ -- see Fig.~\ref{fig:ergo}.

\begin{figure}[t!]
	\centering
	\includegraphics[width=\linewidth]{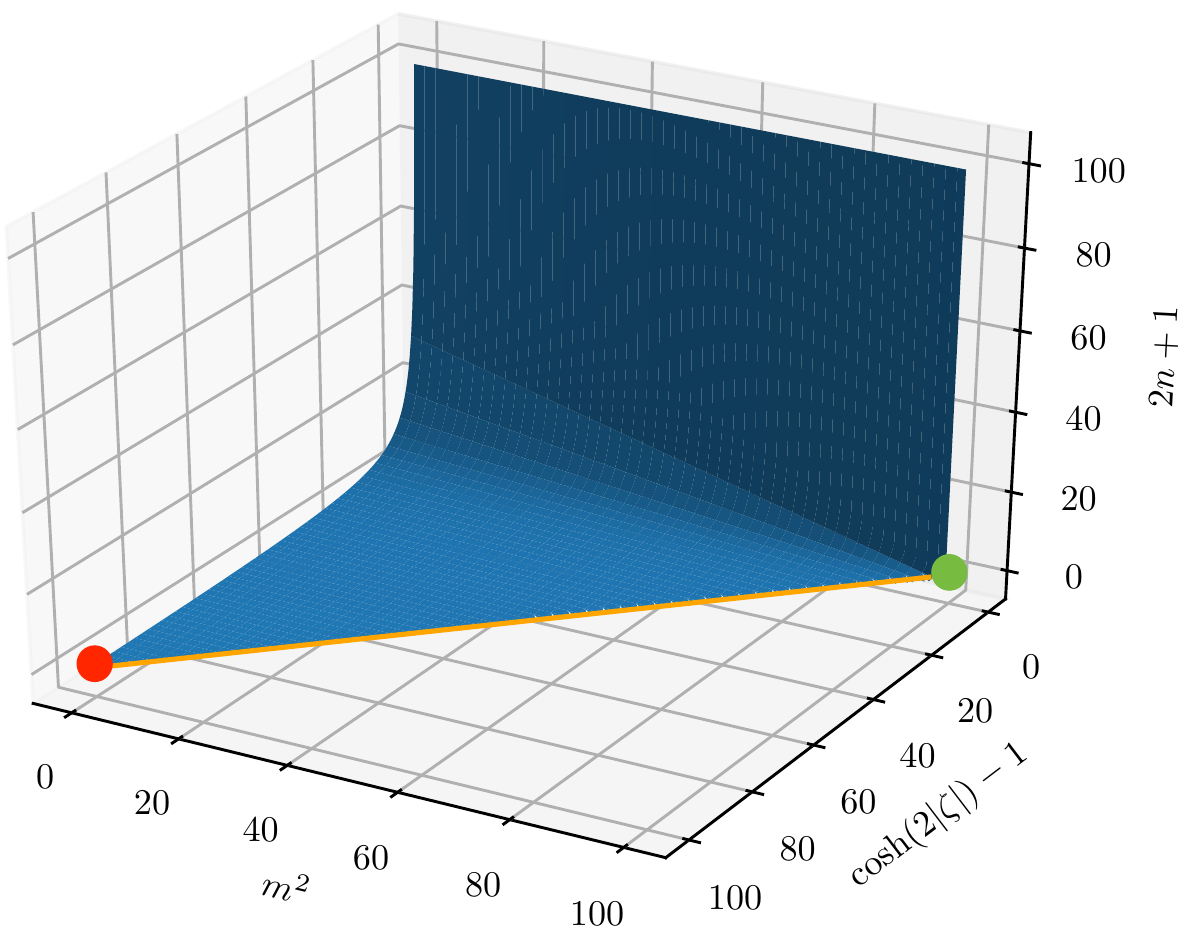}
	\caption{Geometric representation of the iso-ergotropic 2D manifold associated with Gaussian states
	~(\ref{PARGAUS}) with  ${\cal E}(\hat{\rho})=50$ in the phase space defined by the parameters $x:=m^2$, $y:=(\cosh(2|\zeta|)-1)$,
	$z:= 2{{{n}}}+1$. The green dot represents the pure coherent element of the set
	($y=0$, $z=1$), while the red dot the pure squeezed element ($x=0$, $z=1$). The orange line at the bottom represents
	the pure states of the manifold.}
	 \label{fig:ergo}
\end{figure}

\section{The playing field}\label{sec:PlayingField}

Motivated by the Kelly Criterion for optimal betting and the versatility of the one-mode BGCs, we present here a theory of classical betting with a quantum payoff. 
First, in Subsection~\ref{sec:KellyCriterion} we briefly review the conventional Kelly setting. Then in 
Subsection \ref{sec:theProtocol}, we introduce the repeated-quantum-betting scheme as the iteration of single quantum bets, each bet being a probabilistic transmission of the quantum state of a single Bosonic mode through a one-mode BGC. We hence embed the  protocol into a rigorous mathematical framework. Not only this allows for a more compact representation of the procedure over time, but also it enables to see the iterative scheme from the higher perspective of random dynamical systems: the action of a family of quantum transformations, corresponding to one-mode BGCs in the procedure, over a phase space of (single Bosonic mode) quantum states. Finally in Subsection \ref{sec:KellyQuantumClassical} we construct a parallel between the quantum and the classical Kelly criterion.

\subsection{The Kelly Criterion for optimal betting}\label{sec:KellyCriterion}

Imagine that one has $J$ horses competing in a race, each
characterized by a winning probability $\{ p_j\}_{j=1,\ldots,J}$,
$\sum_{j = 1}^{J} p_j= 1$.
Formally speaking the $p_j$'s define a distribution   
of a random variable $Z$ taking values in the alphabet $\left\{1, \ldots, J\right\}$ with probability $\mathcal{P}(Z = j) = p_j$ (symbols $j$ of the alphabet representing the horses whereas the  variable $Z$ the winning horse).
If a gambler invests one dollar on horse $j$ then they receive $o_{j}$ dollars if horse $j$ wins and zero dollars if horse $j$ loses: the game is then said to be under {\it unfair odds} condition if $\sum_{j = 1}^{J} \frac{1}{o_{j}}>1$; on the contrary, if $\sum_{j = 1}^{J} \frac{1}{o_{j}}=1$ or $\sum_{j = 1}^{J} \frac{1}{o_{j}}<1$ we say that the odds are 
{\it fair } or  {\it super-fair} respectively. In what follows we shall always focus on these last two cases for which one can show that the best gaming option for the gambler is always to distribute all of their wealth across the horses (constant allocation strategy). Accordingly we call $b_{j}$ the  fraction of the gambler's wealth invested in   horse $j$, $b_{j} \geq 0$ and $\sum_{j = 1}^{J} b_j = 1$. We consider henceforth a sequence of $t$ repeated gambles on this race under 
{\it i.i.d.} assumptions, i.e. requiring that neither the probabilities $p_j$s,
 the odds $o_j$s, and  the fractions $b_j$s vary from betting stage to betting stage.   
The wealth of the gambler during such a sequence is a random variable defined by the expression 
\begin{eqnarray} \bar{S}_{t} \label{DEFST}
:= S_{j_t} \cdots  S_{j_2} S_{j_1} \;, 
\end{eqnarray}
where for  $\ell\in \{ 1,\cdots, t\}$, $j_\ell\in \left\{1, \ldots, J\right\}$ represents the horse that has won the
$\ell$-th race delivering a payoff
$S_{j_\ell} = o_{j_\ell}  b_{j_\ell}$, the corresponding joint 
probability being 
\begin{eqnarray}  \label{PROBev}  
\bar{p}_{t}:= p_{j_t}\cdots
p_{j_2} p_{j_1}\;,\end{eqnarray} 
(to help readability hereafter we adopt the compact
convention  of dropping explicit reference of the functional dependence
upon the indexes ${j_t}$,  $\cdots$, ${j_2}$, ${j_1}$ of  quantities evaluated along a given stochastic trajectory). In the asymptotic limit of large $t$,  the value of $\bar{S}_t$ can be directly linked to the \textit{doubling rate} of the horse race, i.e. the quantity 
\begin{eqnarray} \label{DOUBLING} 
W(\mathbf{b}, \mathbf{o}, \mathbf{p}):= 
\sum_{j = 1}^{J}  p_j {\log_2}(o_j b_j)  \;,
\end{eqnarray}  
where $\mathbf{p}:= (p_1,p_2,\cdots, p_J)$, $\mathbf{o}:= (o_1,o_2,\cdots, o_J)$ and $\mathbf{b}:= (b_1,b_2,\cdots, b_J)$ are the
three vectors which, under {\it i.i.d.} assumptions fully define the model: the 
first assigning the stochastic character of the race, 
the second the selected odds of the broker, and the last one
the strategy the gambler adopts in placing his bets. More specifically, the connection between $\bar{S}_t$ and 
$W(\mathbf{b}, \mathbf{o}, \mathbf{p})$ is provided by the strong Law of Large Numbers
(LLN) which in the present case establishes that for sufficiently large $t$  we have 
$
\bar{S}_{t} \approx  2^{t\,W(\mathbf{b}, \mathbf{o}, \mathbf{p})}$ almost surely, or more precisely that 
\begin{equation}\label{eq:slln}
{\Pr}\left[\lim_{t\rightarrow\infty}\frac{{\log_2} \bar{S}_{t}
}{t}=W(\mathbf{b}, \mathbf{o}, 
\mathbf{p})\right]=1\;,
\end{equation}
see for instance \cite{cover2012elements}, Theorem 6.6.1.
Exploiting this result, one can now identify the optimal gambling strategy that ensures the largest asymptotic wealth increase, by selecting  $\mathbf{b}$ which allows $W(\mathbf{b}, \mathbf{o},\mathbf{p})$ to reach its maximum value $W^{*}(\mathbf{o},\mathbf{p})$ for fixed $\mathbf{o}$ and $\mathbf{p}$. Such  maximum is achieved by the proportional gambling scheme where $\mathbf{b}= \mathbf{p}$ (Kelly criterion)  so that 
\begin{equation} W^{*}(\mathbf{o},\mathbf{p}) = \sum_{j = 1}^{J} p_j {\log_2} (o_jp_j) =
\sum_{j = 1}^{J} p_j {\log_2} o_j -  \mathcal{H}(\mathbf{p})\;, 
\end{equation}  with $\mathcal{H}(\mathbf{p}):=\sum_{j=1}^J p_j \log_2 p_j$ the Shannon entropy of the Bernoulli distribution defined by the vector $\mathbf{p}$ (see \cite{cover2012elements}, Theorem 6.1.2).

\subsection{BGC betting model}\label{sec:theProtocol}

As in the previous section we consider a horse-race model
where at each step of the gambling scheme $J$  horses
compete, their probability of winning being defined by  a Bernoulli distribution $\{ p_j\}_{j=1, \cdots, J}$. At variance with the typical Kelly setting
we assume that the capital the gambler invests in the game
is represented not by conventional money, but instead by the ergotropy~(\ref{ERGOTROPY}) they store into a single Bosonic mode $A$ initially prepared into the density matrix $\hat{\rho}_0$ which, as anticipated in the introduction, is assumed to belong to the Gaussian subset.

To place their bet, at each race the gambler
divides  his capital exploiting an optical splitter analogous to one shown in Fig.~\ref{fig:BS} 
which distributes the input mode $A$ into $J$ independent  outputs characterized by effective transmissivities $\eta_j\in ]0,1[$, each associated
with an  individual horse and fulfilling the normalization condition
\begin{eqnarray} 
\sum_{j=1}^J \eta_j^2 =1\;. \label{NORM} 
\end{eqnarray} 
This implies that the capital
the gambler is placing on 
the horse $j$ is locally stored 
into the density matrix obtained by applying the purely lossy {BGC} mapping $\Phi_{\eta_j}^{(\text{loss})}$ of Eq.~(\ref{LOSSY})  to the entry capital
of the betting stage. The parameters $\eta_j$  clearly reflect the gambler game strategy  and play the role of the fractions $b_j$ of the original Kelly's model (more precisely, due to constraint ~(\ref{NORM}) 
the exact correspondence is between the $b_j$s and the squares  $\eta_j^2$ of the network transmissivities).  In case the $j$-th horse wins, the corresponding output port will hence undergo a coherent amplification event which pumps up the energy level of the
system via the gain parameter $k_j > 1$ and is returned to the gambler, the other sub-modes being instead absorbed away.  At the level of $A$ this results in a transformation 
described by the pure amplifier {BGC} mapping $\Phi_{k_j}^{(\text{amp})}$ introduced in Eq.~(\ref{AMPY}).  The parameters $k_j$ reflect the odds parameters $o_j$ of the classical model.  Again the exact correspondence is between the $o_j$s and the squares 
$k^2_j$ of the amplification amplitudes which allows us to translate the {\it fair odds}  and {\it super-fair odds} conditions  into  the requirements 
\begin{eqnarray}
\sum_{j = 1}^{J} \frac{1}{k^2_{j}}&=&1\;,  \qquad (\mbox{fair odds})   \label{FAIR}
\\ \sum_{j = 1}^{J} \frac{1}{k^2_{j}}&<&1\;, \qquad (\mbox{super-fair odds})\label{SFAIR} 
\end{eqnarray}   respectively, which we shall assume hereafter.

In summary, after each betting stage the gambler 
receives back the state of their single mode resource quantum memory,
transformed via the combined action of a purely lossy map and a pure amplifier, i.e. the BGC map which according to~(\ref{COMPO1}) can be written as
\begin{equation} \label{COMPO2} 
\Phi^{\text{}}_{{j}} := \Phi[{{g_j}};{\alpha_j}]= \Phi_{{k_j}}^{(\text{amp})} \circ \Phi_{{\eta_j}}^{(
\text{loss})},
\end{equation}
where now, for all $j\in \{ 1, \cdots, J\}$, the stretching and noise terms are given by 
\begin{eqnarray}
g_j&:=& k_j \eta_j\;,  \\
{\alpha_j} &:=& \frac{{k_j}^2(1-{\eta_j}^2) + {k_j}^2 -1}{2}\;.
\end{eqnarray}
The composition property ~(\ref{COMPO1}) also
lends itself well to the framework of repeated gambles described in Sec.~\ref{sec:KellyCriterion}. 
\begin{figure}[t!]
	\centering
	\includegraphics[width=0.92\linewidth]{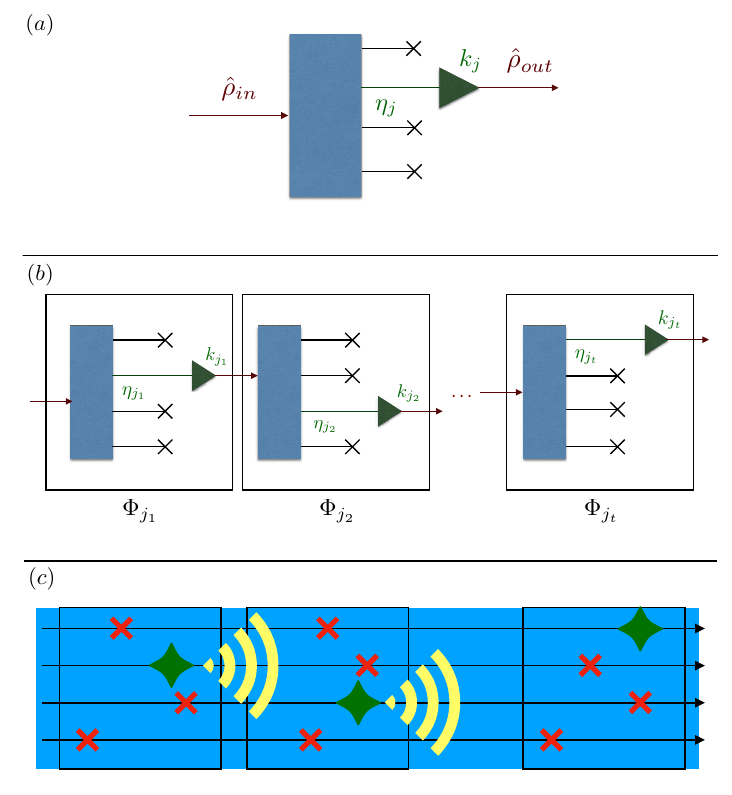}
	\caption{Schematic of the quantum payoff betting protocol:  Panel (a) 
	single betting stage; Panel (b) concatenation of multiple bettings. In the figures 
	the blue rectangles represents optical splitters analogous to the one depicted in Fig.~\ref{fig:BS}; 
	the first red arrow represent the state of $A$ which encode the invested capital whereas the second one is the
	one associated with the winning stake. Black arrows represent the losing fractions of capital, while the winning (green) one $\eta_j$ is \textit{amplified} by the odd $k_j$; Panel (c) schematic representation of the associated lumped-element model scheme describing the
	propagation of electromagnetic modes inside a multi-layer material (blue area): here the red crosses represent complete photon absorption, the
	the green diamond stimulated emission process, and the yellow elements multi-modes scattering events. The resulting scheme closely reminds us of a general random lasing process where light is randomly amplified (and in part absorbed) by a disordered material while maintaining a certain degree of quantum coherence~\cite{Wiersma2008}.
	  }
	 \label{fig:ATTAMP}
\end{figure}
Specifically, assuming the initial state of $A$ to be described by the density matrix $\hat{\rho}_{0}$, after the fist betting event, with probability $p_{j_1}$, the gambler will get the state
\begin{eqnarray}\hat{\rho}_{0}  \longrightarrow  \hat{\rho}_{1}:=\Phi^{\text{}}_{j_1}(\hat{\rho}_{0})\label{rho0torho1}\;,
\end{eqnarray} 
which they will use to place the second bet
(notice that following the same convention introduced in the previous
section, the functional dependence of  $\hat{\rho}_1$ 
upon the index  $j_1$ is left implicit).
Accordingly, with conditional probability $p_{j_2}$ the second horse
race will induce the following mapping on~$A$ 
\begin{eqnarray} \hat{\rho}_{1}  \longrightarrow 
\hat{\rho}_{2}:=\Phi^{\text{}}_{j_2}( \hat{\rho}_{1}) \label{rho0torho2}\;.
\end{eqnarray} 
At the level of input density matrix $\hat{\rho}_0$, Eq.~(\ref{rho0torho2}) 
corresponds to the stochastic mapping 
\begin{eqnarray} \hat{\rho}_{0}  \longrightarrow  \hat{\rho}_{2}= 
\bar{\Phi}_2 (\hat{\rho}_{0}) \label{TWOSTAGE} 
\label{rho0torho21}\;,
\end{eqnarray} 
where 
\begin{eqnarray} \label{TWOEVE} 
\bar{\Phi}_2
:=    \Phi^{\text{}}_{j_2}\circ
\Phi^{\text{}}_{j_1}\;,
\end{eqnarray} 
is the composite BGC map obtained by merging together $\Phi^{\text{}}_{j_2}$ and 
$\Phi^{\text{}}_{j_1}$ and occurs with joint probability
$\bar{p}_2
=p_{j_2} p_{j_1}$.  More generally,  indicating with $\hat{\rho}_{j_1,\cdots, j_{t}}$ 
the state of $A$ after $t$ steps along a given stochastic realization of the horse race described by a winning string of event $j_1, j_2, \cdots, j_t$, 
the  transformation~(\ref{rho0torho2}) gets replaced by the mapping 
\begin{eqnarray}
\hat{\rho}_{t-1} \longrightarrow \hat{\rho}_{t}:= \Phi^{\text{}}_{j_t}(
\hat{\rho}_{t-1})
\label{rho0torhot}\;,
\end{eqnarray}  
which occurs with probability  $p_{j_t}$.
Similarly the transformation (\ref{TWOSTAGE}) becomes 
\begin{eqnarray}  
\hat{\rho}_{0}  \longrightarrow  \hat{\rho}_{t}:= \bar{\Phi}_{t}(\hat{\rho}_{0})\label{rho0torhotT}\;,
\end{eqnarray}  
where now the $t$-fold mapping  $\bar{\Phi}_{t}$
is given by 
\begin{eqnarray}
\bar{\Phi}_{t}:= \Phi^{\text{}}_{j_t}\circ \cdots \circ
\Phi^{\text{}}_{j_2} \circ \Phi^{\text{}}_{j_1}
\label{rho0torhotTT}\;,
\end{eqnarray}  
and occurs with joint probability $\bar{p}_t$  defined as in 
Eq.~(\ref{PROBev}) -- the whole process being visualized in Fig. \ref{fig:ATTAMP}.
We notice that from Eq.~(\ref{COMPO2}) and the general
composition rule~(\ref{COMPO1}) it follows that
\begin{equation} \label{COMPO33} 
\bar{\Phi}_{t} =  \Phi[\bar{g}_t;\bar{\alpha}_t] \;, 
  \end{equation}
where the parameter 
$\bar{g}_t$ is given by  
\begin{eqnarray} \bar{g}_t \label{DEFBARGT} 
:= g_{j_t} \cdots  g_{j_2} g_{j_1} \;, 
\end{eqnarray}
while $\bar{\alpha}_t$ is defined by 
the recursive formula 
\begin{eqnarray} \label{DISTRALPHA} 
\left\{ \begin{array}{l}
\bar{\alpha}_{t=1}
:= {\alpha}_{j_1} \;, \\ \\
\bar{\alpha}_t= g_{j_t}^2\,\bar{\alpha}_{t-1} + \alpha_{j_t}\;,
\end{array} \right. 
\end{eqnarray}
which, as shown in Lemma~\ref{lem:recursive} of the  Appendix, admits the following solution
\begin{equation}\label{eq:recursivesolutionalphaTESTO}
\bar{\alpha}_t 
= \bar{g}^2_t \sum_{\ell = 1}^{t} \frac{\alpha_{j_\ell}}{\bar{g}_\ell^2}.
\end{equation}
Before proceeding any further it is worth stressing that  the map $\bar{\Phi}_{t}$ defined above
describes the state of  $A$ after $t$ steps along a given trajectory of the betting process.
Taking the weighted mean of these super-operators with probabilities  $\bar{p}_t$ gives instead the transformation which defines the  evolution of the system irrespective of the history of the betting process. Explicitly this is given by 
\begin{equation}  \label{AVEMAP} 
\bar{\Psi}_{t}:=  \sum_{j_1,\cdots,j_t} p_{j_1} \cdots p_{j_t} {\Phi}_{j_1}\circ \cdots\circ {\Phi}_{j_t}
= \underbrace{\Psi\circ \cdots \circ \Psi}_\text{$t$-times}  \;, \end{equation}
with  $\Psi$ the CPTP map defined by the convex convolution  
\begin{eqnarray}
\Psi:= \sum_{j=1}^J  {p}_j {\Phi}_{j}\;, \end{eqnarray} 
which does not necessarily correspond to a BGC. 
It is finally worth stressing that the scheme  provides an effective, lumped-element model description of 
propagation of light modes inside a multi-layer material where, depending on the layer, they get randomly amplified (as in a stimulated emission process) or attenuated (as in an absorption process), the  stochastic channel  $\bar{\Phi}_{t}$ describing the system evolution conditioned to the realization of a specific history of events -- see panel (c) of Fig.~\ref{fig:ATTAMP}. As a matter of fact our setting provides an effective model for random lasing~\cite{Wiersma2008}. A random laser is a device where coherent light amplification is achieved without the presence of an optical cavity, through a sequence of stochastic events that may occur in disordered materials~\cite{markushev, Lawandy1994} due to multiple-scattering effects~\cite{letokhov, PhysRevE.54.4256}. 
At the level of the the electromagnetic radiation, such processes will lead to trajectories which 
ultimately can be characterized as multi-mode versions of  the maps $\bar{\Phi}_{t}$ introduced in Eq.~(\ref{rho0torhotTT}), with the channels $\Phi^{\text{}}_{j}$ describing the gains and losses experienced by 
the signal in their random walk through the active media.

\subsection{Quantum payoff  of the BGC model}\label{sec:payoff}
From the above results it is now easy to express 
the energy  associated with the state of $A$ at a generic step $t$ along
a given trajectory. Indeed  from (\ref{defm}) and (\ref{defsigma}) it follows that 
the first order expectation vector $\bar{m}_t $ and the covariance matrix $\bar{\sigma}_t$ of $\hat{\rho}_{t}$ can be written as
\begin{eqnarray} \label{defmnew} 
\bar{m}_t & =&\bar{g}_{t} \; m \;, \\
\label{defsigmanew} 
\bar{\sigma}_t& =& \bar{g}_{t}^2 \sigma + 2 \bar{\alpha}_{t}\; \mathbb{I}\;,
\end{eqnarray}
where $m$ and $\sigma$ are their corresponding counterparts associated with the input 
state $\hat{\rho}_{0}$. 
Therefore, the mean energy of the system is now given by 
\begin{equation}  
\bar{E}_t:=E(\hat{\rho}_{t})=\bar{g}_{t}^2  E_0 + \bar{\alpha}_t
= \bar{g}_{t}^2 ( E_0 + \bar{\gamma}_t) \;, \label{DEFENERGnew}
\end{equation}
where we introduced the quantity 
\begin{eqnarray}\label{eq:ell}
\bar{\gamma}_t:= \frac{\bar{\alpha}_{t}}{\bar{g}_t^2} 
= \sum_{\ell = 1}^{t} \frac{\alpha_{j_\ell}}{\bar{g}_\ell^2},
\end{eqnarray}
and where $E_0:= E(\hat{\rho}_0)$ is the mean energy of the input state.
Parametrizing the latter as in (\ref{DEFENERGGAUS}),
Eq.~(\ref{DEFENERGnew}) can then be casted in the form  
\begin{eqnarray}
\bar{E}_t=\bar{g}_t^2\left(\frac{2{{{n}}}+1}{2}\cosh(2|\zeta|)+\bar{\gamma}_t +\frac{m^2}{2} \right)
\;.\end{eqnarray}
Similarly, reminding that $\hat{\rho}_0$ is a Gaussian state, from Eqs.~(\ref{DIFFE}), (\ref{defmnew}), and (\ref{defsigmanew}), 
we can express the ergotropy $\bar{\mathcal{E}}_t :=\mathcal{E}(\hat{\rho}_t)$ of  $\hat{\rho}_t$ as 
\begin{eqnarray} \label{DEFOUTERG}
\bar{\mathcal{E}}_t &=& \bar{E}_t  - \bar{n}_{t} -1/2 \;,
\end{eqnarray} 
with $\bar{n}_t\geq 0$ being the thermal contribution to the total photon number of the final state computed as 
\begin{equation}
\bar{n}_t := \bar{g}_t^2 \frac{\sqrt{(2{{{n}}}+1)^2+4\bar{\gamma}_t(2{{{n}}}+1)\cosh(2|\zeta|)+4\bar{\gamma}_t^2}}{2} -\frac{1}{2}\;. 
\end{equation} 

In our construction Eq.~(\ref{DEFOUTERG}) 
represents the payoff of the gambler after $t$ steps. 
By re-organizing the various contributions, it can be conveniently casted into the form 
\begin{eqnarray} \label{ERGOTRAJECTj} 
\bar{\mathcal{E}}_t 
&=&\bar{g}_t^2 \bigg({\mathcal{E}}_0  - \bar{\Delta}_t\bigg)\;,
\end{eqnarray}
where 
\begin{eqnarray}
{\mathcal{E}}_0:= {\mathcal{E}(\hat{\rho}_0)}= \frac{(2{{{n}}}+1)(\cosh(2|\zeta|)-1) + {m}^2}{2}\;,
\end{eqnarray} 
 is the initial value of the
ergotropy, i.e. the initial capital invested by the player. In the above expression  
\begin{equation} \label{DELTADEF}
\bar{\Delta}_t:= \frac{(2{{{n}}}+1)}{2}\Big({\sqrt{1+4  \bar{\Gamma}_t\cosh(2|\zeta|)+4\bar{\Gamma}_t^2}} -
{(1+2\bar{\Gamma}_t )}\Big)\;, 
\end{equation} 
is a non-negative quantity which, while not being directly linked to ${\cal E}_0$, is an explicit functional of the selected input state $\hat{\rho}_0$ 
 (here $\bar{\Gamma}_t := \frac{\bar{\gamma}_t}{2{{{n}}}+1}$). This fact   
 has profound consequences as it implies that, differently from 
the original Kelly's scheme, in our setting the final payoff  $\bar{\mathcal{E}}_t$ is not just proportional to the initial invested capital ${\cal E}_0$, but also 
depends in a nontrivial way on the
{\it type} of input state $\hat{\rho}_0$ that  was selected to "carry" the value of such capital. Accordingly in discussing
the performance of the betting scheme   we are now facing an extra layer of 
optimization where  the gambler, beside selecting the transmissivities $\eta_j$ and deciding the value of ${\cal E}_0$ they want to invest in the game, has also the possibility of deciding which one of the input states 
$\hat{\rho}_0$ that belongs to the same iso-ergotropy manifold  (see of Fig.~\ref{fig:ergo}) they want to adopt: 
different choices will indeed result in different $\bar{\Delta}_t$ and hence in different payoff values~$\bar{\mathcal{E}}_t$.

\subsection{Input state optimization} \label{sec:payoffinput}

Let start observing that 
from  Eq.~(\ref{ERGOTRAJECTj}) it follows that  if the input state of the system $\hat{\rho}_0$ is a purely Gibbs thermal state of the system Hamiltonian
(i.e. if $m=\zeta=0$) then the capital  gain is strictly null at all time steps, i.e. $\bar{\mathcal{E}}_t=0$ for all $t$: this implies
that in order to  have at a chance of getting some positive payoff, the player needs  to invest at least a  non-zero capital (a very reasonable 
contraint). 
Notice also that if we restrict to a \textit{semiclassical} case in which $\hat{\rho}_0$
is described by a noisy coherent state (i.e. for $\zeta=0$), 
we get $\bar{\Delta}_t=0$ and hence
\begin{eqnarray} \label{COHERENTSTATES} 
\bar{\mathcal{E}}_t=\bar{g}_t^2\frac{m^2}{2} = \bar{g}_t^2 {\mathcal{E}}_0\;,
\end{eqnarray} 
implying that under this assumption the capital evolves as the classical case of  a geometric game.
Departures from Eq.~(\ref{COHERENTSTATES}) can instead be observed  when  
squeezing is present in the initial state of the system paving the way to the optimization problem 
 mentioned at the end of the previous paragraph, i.e.  deciding what is the best choice of $\hat{\rho}_0$ to be employed in the betting 
 scheme for assigned values of $\eta_j$ and~${\cal E}_0$.

To tackle this issue  we find it useful to 
define two different figures of merit, i.e. the renormalized ergotropy functional $\bar{\mu}_t$  which, as $\bar{\mathcal{E}}_t$
gauges the absolute payoff, and
the output ergotropy ratio $\bar{r}_t$, aimed instead to characterize the amount of resources wasted in the process. 
The first quantity is simply obtained by rescaling the payoff capital $\bar{\mathcal{E}}_t$ by the term $\bar g^2_t\mathcal{E}_0$ in order to contrast
the potential divergent behaviour of the former, i.e. 
\begin{eqnarray} \label{defmu}
\bar{\mu}_t := \frac{\bar{\mathcal{E}}_t}{\bar g^2_t\mathcal{E}_0} = 1 - \frac{\bar\Delta_t}{\mathcal{E}_0}\;.
\end{eqnarray}
For fixed choices of $\eta_j$ and ${\cal E}_0$, $\bar{\mu}_t$ exhibits the same  dependence upon $\bar{\mathcal{E}}_t$ 
upon the input state $\hat{\rho}_0$, with the main advantage of being bounded from above by 1 (on the contrary, as we already mentioned, $\bar{\mathcal{E}}_t$ can explode).  
The second quantity instead is defined as 
the ratio between the ergotropy and the mean energy of the state $\hat{\rho}_t$, i.e. 
\begin{eqnarray} \label{DEFERRET}
\bar{r}_t &:=& \frac{ \bar{\mathcal{E}}_t}{\bar{E}_t} = \frac{\bar{\mathcal{E}}_t}{\bar{\mathcal{E}}_t +\bar{{{n}}}_t +1/2}\\
&=& r_0+ (1-r_0) \left( 1 - \frac{\sqrt{1+4\bar{\Gamma}_t\cosh(2|\zeta|)+4\bar{\Gamma}_t^2}}{1+2 (1-r_0)\bar{\Gamma}_t }\right)\;,  \nonumber 
\end{eqnarray}
where 
\begin{eqnarray} 
r_0 &:=& \frac{\mathcal{E}_0}{E_0}= \frac{\mathcal{E}_0}{\mathcal{E}_0 +{{{n}}} +1/2} \;, 
\end{eqnarray} 
 is the value of the ratio  computed on the input state $\hat{\rho}_0$. 
Notice that by construction  
$\bar{r}_t$ and ${r}_0$ are  bounded quantities that cannot 
exceed the value 1.  Furthermore  
 irrespective of the parameters of the model and from the outcome of the stochastic process,  $\bar{r}_t$ cannot be larger than its initial value~$r_0$: indeed 
an upper bound for $\bar{r}_t$ can be obtained by setting equal to zero the squeezing parameter of the initial state (i.e. $\zeta=0$) 
which leads to 
\begin{eqnarray} \label{ddf}
\bar{r}_t \leq  \frac{r_0}{1+2 (1-r_0)\bar{\Gamma}_t } \leq r_0 \;,\end{eqnarray}
the rightmost inequality following as a consequence of the positivity of $\bar{\Gamma}_t$.   
Equation~(\ref{ddf}) implies  that, even though the gambler has a chance of increasing its wealth (i.e the value of the ergotropy of the mode $A$), in the quantum model we are considering here this always occurs at the cost of an inevitable waste of resources.

It goes without mentioning that, as in the case of the payoff capital $\bar{\cal E}_t$, the figures of merit introduced above 
are stochastic quantities which strongly depend
on the random sequence of betting events.
Interestingly enough however   a mere analysis of the functional dependence of $\bar{\mu}_t$, $\bar{r}_t$ upon the input
data of the problem, allows us to draw some important conclusions
even without carrying out a full statistical analysis of the process (a task which,
as we shall see in the next section, is rather complex). For this purpose we first notice that from Eq.~(\ref{COHERENTSTATES}) 
it follows that setting $\zeta=0$, i.e. using input states $\hat{\rho}_0$
 which have no  squeezing, is the proper choice to bust  
$\bar{\mu}_t$ to its upper bound 1. This choice  also ensures the saturation of the first inequality in Eq.~(\ref{ddf}), leading us to an output ergotropy ratio that expressed in terms of ${\cal E}_0$ can be written as 
\begin{eqnarray} \label{ddf11}
\bar{r}_t =  \frac{{\cal E}_0}{{\cal E}_0+ \bar{\gamma}_t + {{{n}}} +1/2  } \leq   \frac{{\cal E}_0}{{\cal E}_0+ \bar{\gamma}_t  +1/2  }  \;,\end{eqnarray}
the last inequality being saturated by setting ${{{n}}}=0$. 
Accordingly we can conclude that, for fixed values of ${\cal E}_0$ and game strategy $\eta_j$,
 the best choice 
 for the input density matrix  $\hat{\rho}_0$ is the pure coherent input state of the selected iso-ergotropy manifold (i.e. the
 green dot of Fig.~\ref{fig:ergo}) since, irrespective of the specific outcomes of the stochastic process that drives the system dynamics, 
  this will ensure
  the saturation of the functional upper bounds for both the two figures of merit we have introduced. 
The above inequalities show the optimality of coherent input states for any $t$ and $\bar{\gamma}_t$. At each step the resulting concatenation of BGCs results in a noisy attenuation/amplification map, so the gambler cannot gain more ergotropy by using a different strategy even with an adaptive one.

\section{Statistical analysis}
\label{sec:KellyQuantumClassical}

In  Kelly's language each bet in our gambling problem is identified with the drawing of a 
Bernoulli random variable and the overall gain/loss of each bet is identified with the parameters  $g_j$ and ${\alpha_j}$. 
In particular, in the new playing field, the analogoue of the gambler's wealth after $t$ races, $\bar{S}_t$ of Sec.~\ref{sec:KellyCriterion}, is the quantity $\bar{\mathcal{E}}_t/\mathcal{E}_0$ of Eq. (\ref{DEFOUTERG}).
As discussed in the previous section the latter  is computed through the action of the BGC $\bar{\Phi}_t$ of Eq. (\ref{COMPO33}) which determines the evolution of the quantum mode $A$ via 
Eq.~(\ref{rho0torhotT}) and which is fully characterized by the parameters $\bar{g}_t$ and 
$\bar{\alpha}_t$. 
As anticipated in  Eq.~(\ref{COHERENTSTATES}), in the absence of input squeezing (i.e. $\zeta=0$) 
the first of these two quantities, i.e. $\bar{g}_t$, represents the amplification of the input signal in the model,
the gambler's wealth being directly proportional 
to $\bar{g}^2_t$. The presence of squeezing in the initial state, makes the  connection between  $\bar{g}_t$ and 
the payoff more complex (see Eq. (\ref{DEFOUTERG}))
still 
a direct comparison between Eq.~(\ref{DEFBARGT})
and  (\ref{DEFST}) makes it explicit that  $\bar{g}_t$ shares the same statistical properties of the function 
$\bar{S}_t$ of the classical setting. In particular also in this case we can use the strong LLN 
to replace Eq.~(\ref{eq:slln}) with 
\begin{equation}\label{eq:sllndd}
{\Pr}\left[\lim_{t\rightarrow\infty}\frac{{\log_2} \bar{g}_{t}
}{t}=G(\boldsymbol{\eta}, \mathbf{k},  
\mathbf{p})\right]=1\;,
\end{equation}
where now 
\begin{eqnarray} \label{DOUBLINGQ} 
G(\boldsymbol{\eta}, \mathbf{k}, \mathbf{p}):= 
\sum_{j = 1}^{J}  p_j {\log_2}(g_j)=\sum_{j = 1}^{J}  p_j {\log_2}(k_j
\eta_j)  \;,
\end{eqnarray}  
substitute the doubling rate $W(\mathbf{b}, \mathbf{o}, 
\mathbf{p})$ of Eq.~(\ref{DOUBLING}) -- here $\boldsymbol{\eta}:= (\eta_1, \cdots,\eta_J)$,  $\mathbf{k}:=(k_1,\cdots, k_J)$ and $\mathbf{p}:=(p_1,\cdots, p_J)$.
Accordingly  we can claim that almost surely
$\bar{g}_t \approx 2^{t G(\boldsymbol{\eta}, \mathbf{k}, \mathbf{p})}$ for $t$ sufficiently large.
Furthermore  simple algebra allows one to 
show that 
the optimal energy splitting strategy which yields the maximum
doubling rate 
\begin{eqnarray} G^*( \mathbf{k}, \mathbf{p}) :=\max_{\boldsymbol{\eta}} G(\boldsymbol{\eta}, \mathbf{k}, \mathbf{p})\;,\end{eqnarray} 
which is obtained by setting 
\begin{eqnarray}
\eta^2_j=p_j\;, \label{QKELLY} 
\end{eqnarray}  
which represents precisely the Kelly criterion for the optimal choice of the betting strategy. 

Let us now turn our attention to the second stochastic quantity which defines the model, i.e. the parameter $\bar{\alpha}_t$.
This term  measures the noise generated in the process due to the random application of quantum BGC, and it appears thanks to the non-commutative nature of our quantum model. From a close inspection of Eq.~(\ref{DISTRALPHA}) we observe that 
$\bar{\alpha}_t$ is generated by the iteration of a family of random Lipschitz maps\cite{bhattacharya2007random}, precisely affine random maps with random slope and random intercept. 
For instance in the case where
the total number of horses in the race is $J=2$ with $p_1=p$, $p_2=1-p$, this can be made explicit by starting
with 
$\bar{\alpha}_0=0$ and writing $\bar{\alpha}_{t} = f(\bar{\alpha}_{t-1})$ with 
\begin{equation} \label{DEFFF} 
f(x):=\begin{cases}g_1^2x+\alpha_1\quad \text{with probability}\quad p\\
g_2^2x+\alpha_2\quad \text{with probability}\quad 1-p\;.
\end{cases}
\end{equation}

Depending on the value of $G^*(\mathbf{k}, \mathbf{p})$ we can distinguish different scenarios.
The one that is best characterized at the mathematical level is when the model is contracting on average \cite{nicol2002fine}, i.e. when $G^*(\mathbf{k}, \mathbf{p}) < 0$, which unfortunately is the less interesting case for our model as it represents a betting scheme where the player asymptotically looses all its capital. Under this condition there exists a unique stationary invariant probability measure for the process which can be singular or absolutely continuous \cite{diaconis1999iterated}. For instance, referring to the parametrization introduced in Eq.~(\ref{DEFFF}), when $p=\frac{1}{2}$ in the case $\frac{1}{2}<g_1^2=g_2^2<1$ if $\alpha_1\not=\alpha_2$ then, up to an affine map, the invariant measure of the system is a Bernoulli convolution \cite{gouezel2019entropic}. These measures have been studied since the 1930’s, revealing connections with harmonic analysis, the theory of algebraic numbers, dynamical systems and fractal geometry. It is well-known that in some cases they are singular (more precisely when the inverse of the contraction rate $g_2=g_1$ is an especially "symmetric" type of irrational algebraic number, a so-called  Pisot number \cite{pisot}). When $g_1^2=g_2^2=\frac{1}{2}$ instead, the measure is the Lebesgue measure whereas if $g_1^2=g_2^2<\frac{1}{2}$ it is the uniform measure on a Cantor set (i.e. a totally disconnected closed subset of the real line consisting entirely of boundary points). Such sets may have zero or positive Lebesgue measure. They are typically fractal sets and the uniform measure is equivalent to the the Hausdorff measure of dimension equal to the dimension of the Cantor set (see \cite{falconer} for more details).
Generally speaking in the contractive-on-average  case the fundamental fact which one can exploit is that the random family of maps contracts the Wasserstein metric on compactly supported probability measures. If this assumption is not met then the situation is considerably more complex: orbits can be dense on the real half-line \cite{misiurewicz2005real} and one can investigate to what extent their distribution is uniform \cite{bergelson2006affine}. 

The situation becomes even more problematic in the case 
 which is more interesting for us, i.e. for maps that  are
   expanding-on-average, i.e. when $G^*(\mathbf{k}, \mathbf{p}) > 0$.  
Not much attention has been devoted to the study of this scenario: as a matter of fact,  the only reference we could find\cite{demichel2018renormalization} investigates only the orbits of the associated topological dynamical systems obtained by suitably renormalizing the $n$--th iterates.
We contribute to the effort by  noticing that in these cases due to the fact that  $\bar{g}_t$ is almost surely  a divergent quantity (see e.g. (\ref{eq:sllndd})),
it is convenient to focus on the renormalized 
version $\bar{\gamma}_t$  of $\bar{\alpha}_t$ introduced in Eq.~(\ref{eq:ell}), which
determines the stochastic evolution of 
indicators $\bar{r}_t$ and $\bar{\mu}_t$ defined in Sec.~\ref{sec:PlayingField}.
It turns out that $\bar{\gamma}_t$ is rather regular and for finite $t$,  its first and second moments
can be computed (see Lemma~\ref{lem:momenta} of the Appendix).
Furthermore it is possible to show that the random parameter $\bar{\gamma}_t$
converges almost surely whenever $\bar{g}_t^2$ diverges exponentially almost surely:
this is a consequence of the fact that $\bar{\gamma}_t$ is $t$-th 
partial sum of the infinite series with $t$-th term $\alpha_{j_t}/\bar{g}^2_t$ 
that always exists due to the non-negativity of the terms, and it is guaranteed to be finite when the denominator can be upper bounded by an exponentially diverging term, as it gives a trivial upper bound by a convergent geometric series (see the Theorem~\ref{convlt} of the Appendix for a more detailed derivation of this result). 
We also observe that, since both $\bar{\mu}_t$ and  $\bar{r}_t$  are positive quantities  bounded from above, 
 it makes sense to determine  the asymptotic values $r$ and $\mu$ as the (point-wise) limit of the distribution $\bar{r}_t$ for $t\rightarrow\infty$:
\begin{equation} \label{DEFERRE} 
\mu := \lim_{t\rightarrow\infty}\bar{\mu}_t \qquad r := \lim_{t\rightarrow \infty}\bar{r}_t.
\end{equation}
Despite this nice property, determining the effective distribution of $\mu$ and $r$ (or 
those of $\bar{\mu}_t$ and  $\bar{r}_t$ for $t$ finite), remains however 
a rather complex task and at present we are not able to make general claims.  To compensate for this, 
we do however present a numerical analysis that allows us to at least get some 
 insights on the problem.

\subsection{Numerics}\label{sec:Numerics}

In this section we simulate numerically the  evolution of $\bar{\mu}_t$ and $\bar{r}_t$ over a sufficiently long time horizon to infer their asymptotic behaviour. 
In Figs.~\ref{fig:SQUEEZED}--\ref{fig:IRREG} the temporal evolution of 
the empirical distributions of these quantities have been reconstructed by sampling $10.000$ simulations for each temporal step and 
for various choices of the parameters and in the two-horse case (i.e. $J=2$).
In the  plots we also report  the temporal evolutions of the
average values of $\bar{r}_t$ and $\bar{\mu}_t$ extracted from the sampling. It goes without mentioning that due to the nonlinear dependence of these functionals  
upon the state of the system, they do not coincide with
the corresponding values computed on the average density matrix of the system obtained by applying the average map 
(\ref{AVEMAP}) to $\hat{\rho}_0$. 
In all the figures we also exhibit  an educated guess for the average
value of $\bar{r}_t$  obtained by naively replacing the  $\bar{\Gamma}_t$ terms that appears in Eq.~(\ref{DEFERRET})
with its average value $\mathbb{E}[\bar\Gamma_t]$ computed as in the Appendix.
This quantity  can be thought of as a \textit{mean field} approximation of our random variable and it is written as
\begin{equation} \label{DEFERREFALSO}
\tilde{r}_t := 1- \frac{\sqrt{1+4\mathbb{E}\left[\bar\Gamma_t\right]\cosh(2|\zeta|)+ 4\mathbb{E}\left[\bar\Gamma_t\right]^2}}{\cosh(2|\zeta|) + 2\mathbb{E}\left[\bar\Gamma_t\right]+F},
\end{equation}
with 
\begin{eqnarray}
\mathbb{E}\left[\bar\Gamma_t\right] =
\frac{1-\mathbb{E}\left[1/g^2\right]^{t}}{1- \mathbb{E}\left[1/g^2\right]}\; \frac{ \mathbb{E}\left[{\alpha}/{g^2}\right]}{2{{{n}}} + 1}\;,
\end{eqnarray} 
converging to the asymptotic value $ \frac{1}{1- \mathbb{E}\left[1/g^2\right]}\; \frac{ \mathbb{E}\left[{\alpha}/{g^2}\right]}{2{{{n}}} + 1}$.

Let us now enter into the details of our numerical analysis.  In  Fig.~\ref{fig:SQUEEZED}, focusing on the same betting scenario,  we compare the results obtained when adopting the optimal (Kelly) betting strategy~(\ref{QKELLY}) 
  (panels (a) and (b)) with those of a sub-optimal one (panels (c) and (d)), for a pure squeezed input state, under super-fair odds conditions. As clear from the plots the Kelly strategy leads to a beneficial increment both  
 in terms of  the mean values of  $\bar{r}_t$ and $\bar{\mu}_t$ (attaining higher values in the (a) and (b) panels), as well as in terms of their spreads (more concentrated  toward higher values in the (a) and (b) panels). 
Similar findings are confirmed  in Fig.~\ref{fig:coherent} were instead we focus on the case of a pure coherent state 
(in this case we only present $\bar{r}_t$, since one always get  $\bar{\mu}_t =1$ for all $t$ due to Eq.~(\ref{COHERENTSTATES})). 
A direct comparison between the panels (a) of the Figs.~\ref{fig:SQUEEZED} and \ref{fig:coherent}, also show that 
under the same betting conditions, and input ergotropy resources, 
the coherent inputs provide better performances than  pure squeezed states in agreement with the results of Sec.~\ref{sec:payoffinput}.
This fact is also confirmed by the data of Fig.~\ref{fig:EFF_DIAG} which
present a numerical estimation of the asymptotic averaged values of $\mu$ and $r$ of Eq.~(\ref{DEFERRE})
for various choices of pure (displaced-squeezed) input state while maintaining the same (optimal) betting strategy. The detrimental influence of thermal photons in the input state 
is instead addressed in  Fig.~\ref{fig:NOISY} where  the temporal distribution of  $\bar{r}_t$ is studied under optimal betting strategies and the same game conditions
discussed in the panels (a) of Figs.~\ref{fig:SQUEEZED} and~\ref{fig:coherent}: 
despite exhibiting the same input ergotropy  one notices that the range and the mean value 
exhibited by $\bar{r}_t$ get contracted with respect to both the coherent and squeezed case. 

\begin{figure} 
	\centering
	\includegraphics[width=250px]{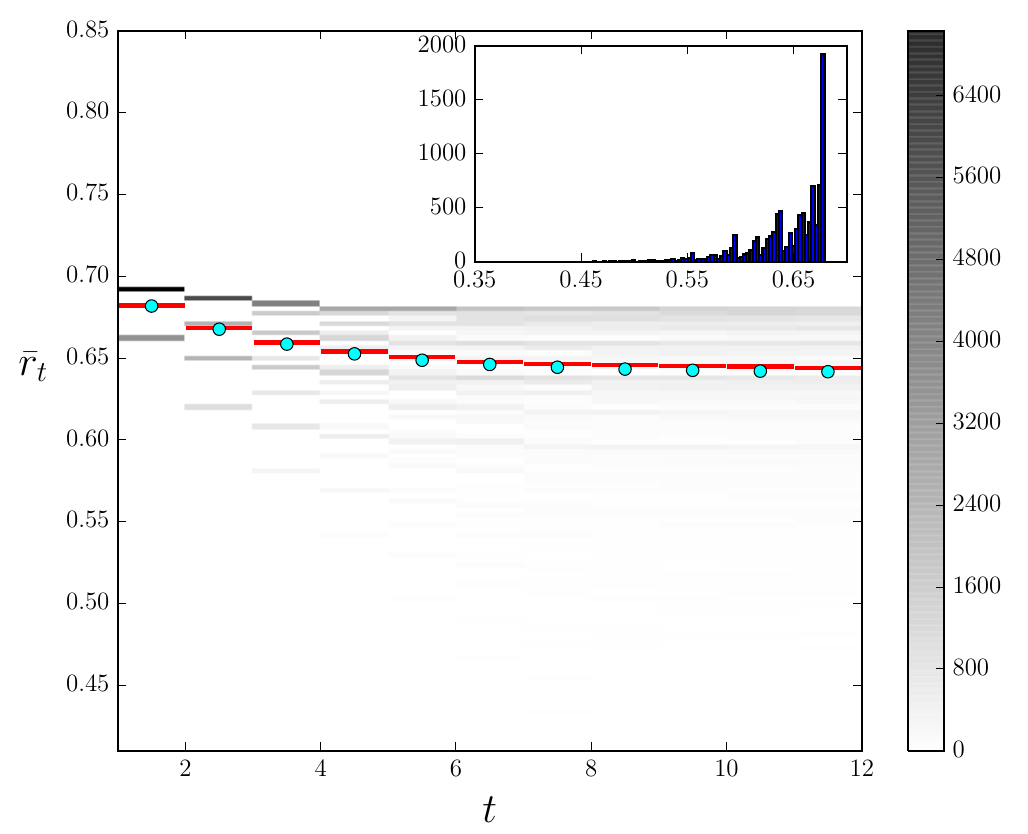}
	\caption{Time evolution of the empirical distribution of $\bar{r}_t$ of Eq.~(\ref{DEFERRET})  for a two horses game ($J=2$), under super-fair odds assumptions
		(\ref{SFAIR}) defined by the vector $\mathbf{k}=(\sqrt{3},\sqrt{3})$, 
		and thermal displaced (zero squeezing) density matrix 
		($m^2=50$, $\zeta=0$, ${{{n}}}=10$) with ergotropy value ${\cal E}_0=25$.
		Here we assumed optimal betting strategy~(\ref{QKELLY}) obtained by  setting $\mathbf{p}= (0.7,0.3)$ and 
		$\boldsymbol{\eta}=(\sqrt{0.7},\sqrt{0.3})$.
		As in Fig.~\ref{fig:SQUEEZED} and \ref{fig:coherent},  the data were obtained over samples of $10.000$ simulations.
		The red lines in the plots represent  the empirical average of the sample while the cyan dots are the mean field approximation $\tilde r_t$ of Eq.~(\ref{DEFERREFALSO}). 
		The insets show		the asymptotic histograms of the associated values of $\bar{r}_t$
		sampled at $t=100$.
	}
	\label{fig:NOISY}
\end{figure}

As a general remark we point out  that the shape that we observe for the asymptotic distributions of  $\bar{r}_t$  and $\bar{\mu}_t$  are a direct consequence of the 
underlying binomial process which generates a random walk on $[0,1]$ with steps whose length and direction at time $t+1$ depend on the position at time $t$\cite{evertsz1992multifractal}.
We highlight also that when choosing $\boldsymbol{\eta}$ as the optimal Kelly-betting strategy~(\ref{QKELLY}) and $\bf k$ to be a vector of super-fair odds (\ref{SFAIR}), then 
we have finiteness of the first moment of $\bar{\gamma}_t$ -- see panels (a) and (b) of Figs.~\ref{fig:SQUEEZED}--\ref{fig:EFF_DIAG}. 
The same stability condition has been enforced also for all the other configurations we have considered, apart from the case reported in 
Fig.~\ref{fig:IRREG}: here  choosing $\boldsymbol{\eta}$
 as the optimal Kelly-betting strategy~(\ref{QKELLY}) and setting $\bf k$ to be a vector of fair odds (\ref{FAIR}),  while all the moments of $\bar{r}_t$ remain finite (because it has support contained in the interval $[0,1]$),
one has $\lim_{t\rightarrow\infty}{\mathbb{E}[\bar\gamma}_t] = \infty$ (see Lemma \eqref{lem:momenta} in appendix) which tends to compromise the
convergency of the simulations as evident from the reported data and the matching with the mean field
estimation~(\ref{DEFERREFALSO}).

\begin{figure}
	\centering
	\includegraphics[width=250px]{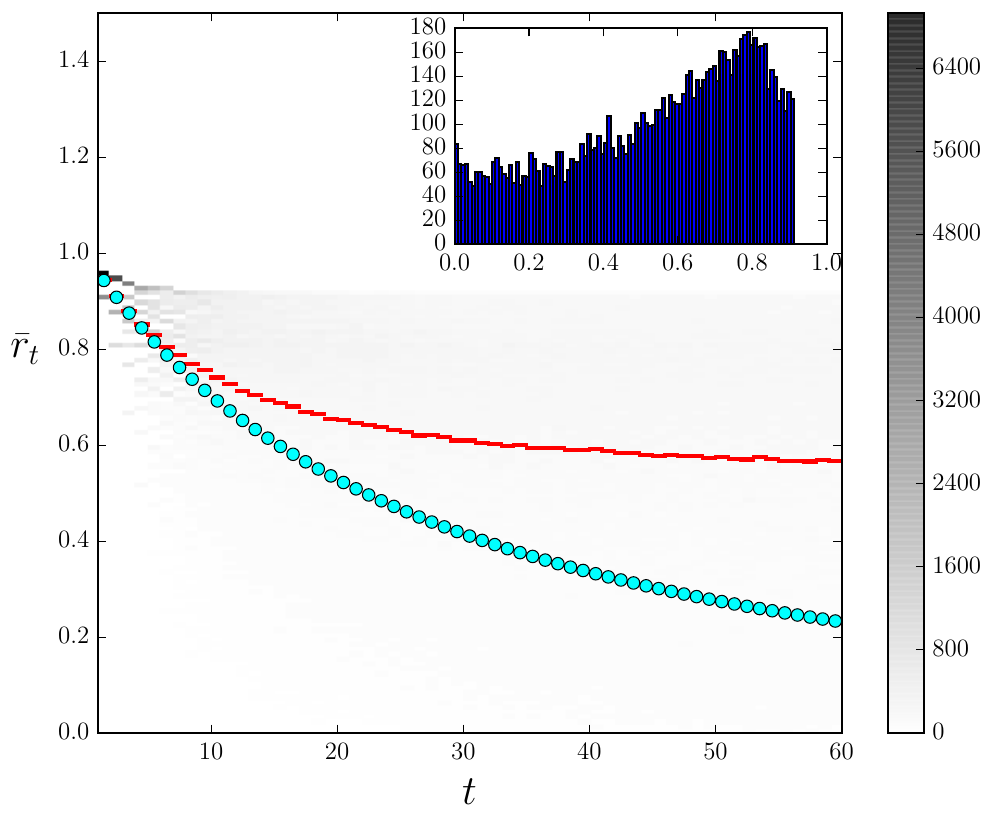}
	\caption{Time evolution of the empirical distribution of $\bar{r}_t$ of Eq.~(\ref{DEFERRET}) over a sample of $10.000$ simulations ($J=2$). The red line is the empirical average of the sample. The cyan dots represent $\tilde r_t$, as we see in this case due to the irregular moments of $\bar\gamma_t$, not only is the convergence slower, but $\tilde r_t$ is also completely different from the sampled mean. The inset figure is the asymptotic histogram, sampled a $t=250$. 
	Here we assumed fair odds condition~(\ref{FAIR}) 
	$\mathbf{k}=(\sqrt{2},\sqrt{2})$, and betting strategy which is slightly sub-optimal
	setting $\mathbf{p}= (0.7,0.3)$ and 
	$\boldsymbol{\eta}=(\sqrt{0.71},\sqrt{0.29})$. The input is a coherent state with 
	 $m^2=50$, $\zeta=0$, ${{{n}}}=0$. Notice that the convergence rate is slower, with respect to the previous plots, due to the increase of the {characteristic time} of the process that is proportional to $1/G$.}
	\label{fig:IRREG}
\end{figure}

\begin{figure}
	\centering
	\includegraphics[width=250px]{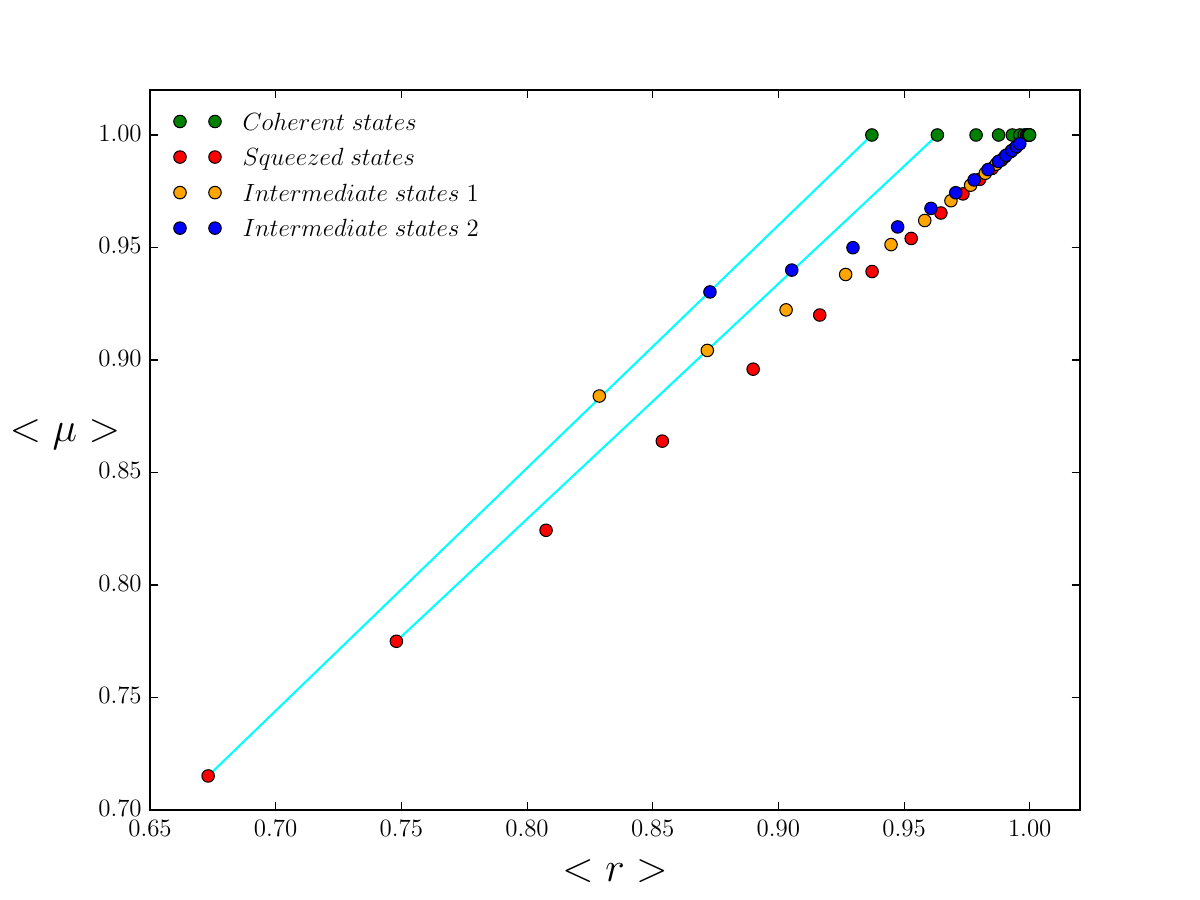}
	\caption{Values of the asymptotic averages for $r$ 
		and $\mu$ of Eq.~(\ref{DEFERRE})   for different types of input state $\hat{\rho}_0$ and for different values of input ergotropy $\mathcal{E}_0$ ranging from $50$ to $50.000$, logarithmically spaced.
		The red circles represent the squeezed input state where $\mathcal{E}_0 = \frac{\cosh(2|\zeta|)-1}{2}$, while the green circles show the coherent input state with $\mathcal{E}_0=m^2/2$. The orange circles show the result for mixed input state where $m^2=3\mathcal{E}_0/4$, finally the blue circles represent the input state where $m^2=7\mathcal{E}_0/8$. All of the aforementioned states have ${{{n}}} = 0$. 
		The cyan lines correspond to the first two \textit{isoergotropic} curves. 
		The quantities $r$ and $\mu$ have been obtained as empirical averages over a sample of $10.000$ simulations for a two horses game $J=2$, with super-fair odds conditions $\mathbf{k}=(\sqrt{3},\sqrt{3})$ (\ref{SFAIR}) and optimal betting strategy~(\ref{QKELLY}) with $\mathbf{p}= (0.7,0.3)$ and $\boldsymbol{\eta}=(\sqrt{0.7},\sqrt{0.3})$. The plots makes it explicit that for fixed value of the initial capital ${\cal E}_0$, 
		optimal perfomances in terms of $\mu$ and $r$ are always attained by using coherent input states.}
	\label{fig:EFF_DIAG}
\end{figure}

\begin{figure*}
	\centering
	\includegraphics[width=\textwidth]{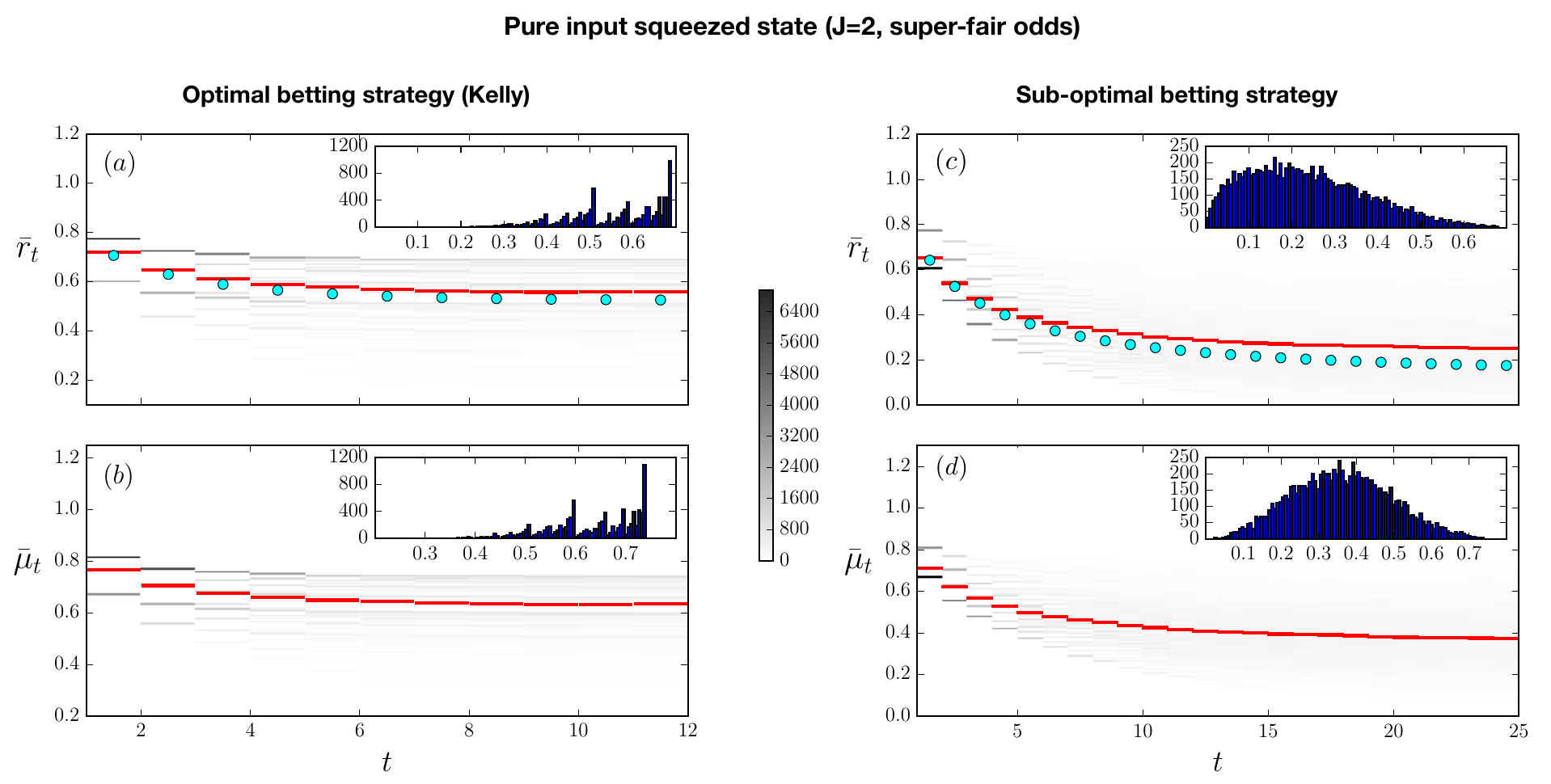}
	\caption{Time evolution of the empirical distributions of $\bar{r}_t$ of Eq.~(\ref{DEFERRET}) and 
		of the (renormalized) ergotropy of the model  $\bar{\mu}_t$  of Eq.~(\ref{defmu}), for a  two horses game ($J=2$) with probability vector $\mathbf{p}= (0.7,0.3)$, under super-fair odds assumptions
		(\ref{SFAIR}) defined by the vector $\mathbf{k}=(\sqrt{3},\sqrt{3})$, and pure squeezed input state of initial ergotropy value ${\cal E}_0=25$  ($m^2=0$, $\cosh(2|\zeta|)=51$, ${{{n}}}=0$).
		Panels (a) and (b): empirical distributions of  $\bar{r}_t$ and $\bar{\mu}_t$ obtained under optimal betting conditions~(\ref{QKELLY}) corresponding to $\boldsymbol{\eta}=(\sqrt{0.7},\sqrt{0.3})$;		Panels (c) and (d):  empirical distributions of  $\bar{r}_t$ and $\bar{\mu}_t$  under sub-optimal betting strategy with  $\boldsymbol{\eta}=(\sqrt{0.3},\sqrt{0.7})$.
		All the data were obtained over samples of $10.000$ simulations.
		The red lines in the plots represent  the empirical average of the sample; the cyan dots are the mean field approximation of the mean $\tilde r_t$ defined in Eq.~(\ref{DEFERREFALSO});
		the insets present instead 
		the asymptotic histograms of the associated values of $\bar{r}_t$ and $\bar{\mu}_t$
		sampled at $t=100$ (panels (a) and (b)), and $t=150$ (panels~(c) and (d)). 
		\label{fig:SQUEEZED}}
\end{figure*}

\begin{figure*}
	\centering
	\includegraphics[width=\textwidth]{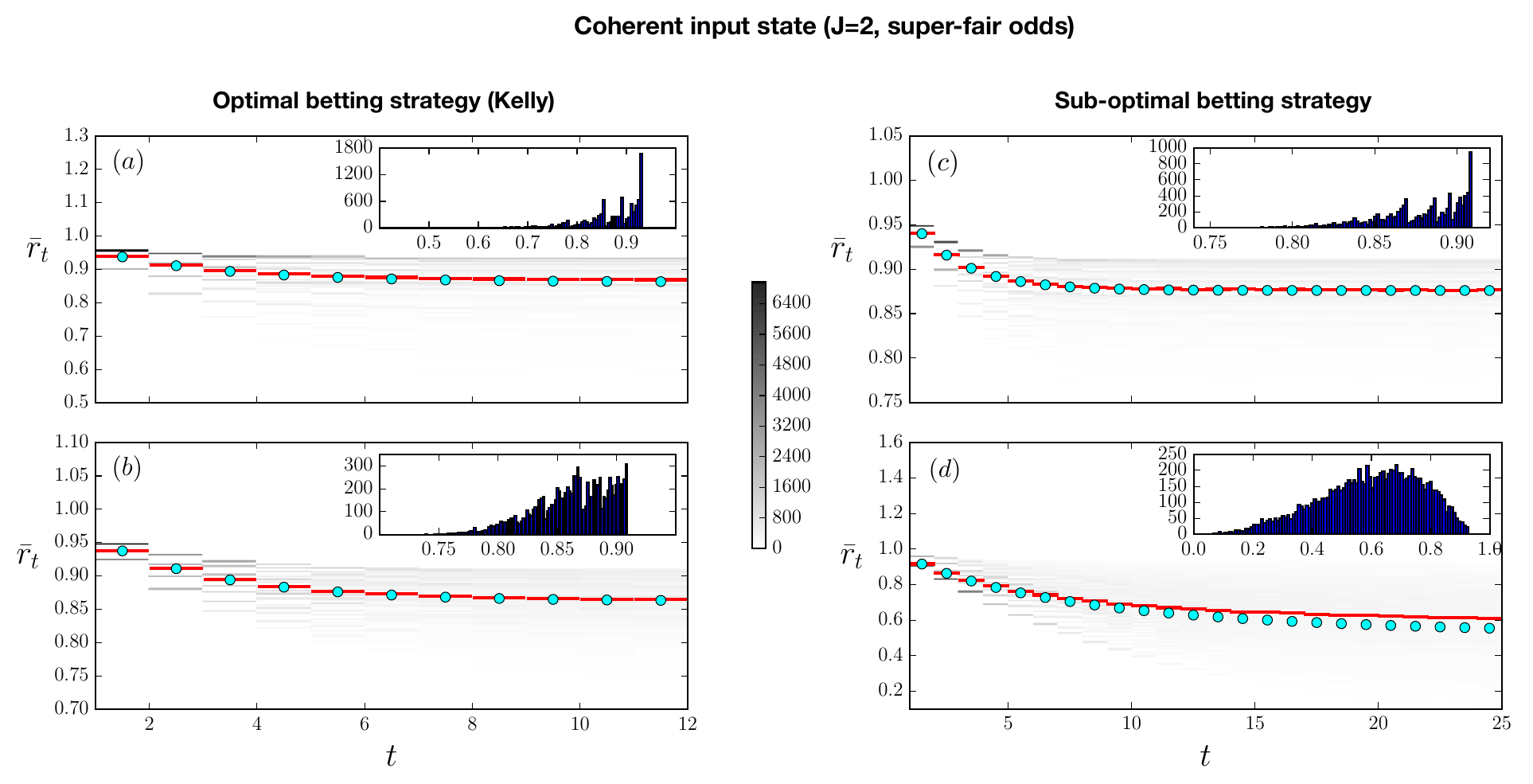}
	\caption{Time evolution of the empirical distributions of $\bar{r}_t$ of Eq.~(\ref{DEFERRET}) for a  two horses game ($J=2$), under super-fair odds assumptions
		(\ref{SFAIR}) defined by the vector $\mathbf{k}=(\sqrt{3},\sqrt{3})$, and pure coherent input state of ergotropy value ${\cal E}_0=25$  ($m^2=50$, $\zeta=0$, ${{{n}}}=0$).
		Panels (a) and (b): empirical distributions of  $\bar{r}_t$  under optimal betting conditions~(\ref{QKELLY})
		obtained by  setting $\mathbf{p}= (0.7,0.3)$, $\boldsymbol{\eta}=(\sqrt{0.7},\sqrt{0.3})$ (for panel (a)),
		and $\mathbf{p}= (0.6,0.4)$, $\boldsymbol{\eta}=(\sqrt{0.6},\sqrt{0.4})$ (for panel (b)).
		Panels (c) and (d):  empirical distributions of  $\bar{r}_t$   under sub-optimal betting strategies
		for the same probability setting of (a) (i.e. $\mathbf{p}= (0.7,0.3)$) but assuming 
		$\boldsymbol{\eta}=(\sqrt{0.6},\sqrt{0.4})$ (panel (c)) and 	$\boldsymbol{\eta}=(\sqrt{0.7},\sqrt{0.3})$ (panel (d)).	 As in Fig.~\ref{fig:SQUEEZED},  the data were obtained over samples of $10.000$ simulations.
		The red lines in the plots represent  the empirical average of the sample while the cyan dots are the mean field approximation $\tilde r_t$ of Eq.~(\ref{DEFERREFALSO}). 
		The insets show		the asymptotic histograms of the associated values of $\bar{r}_t$  sampled at $t=100$ (panel (a) and (b)), and $t=150$ (panels~(c) and (d)).
		It should be stressed in all cases  one has  
		$\bar{\mu}_t =1$ for all $t$ due to Eq.~(\ref{COHERENTSTATES}). }
	\label{fig:coherent}
\end{figure*}

\section{Conclusions and future perspectives}\label{sec:Conclusions}
\noindent Betting (or gambling) is a practical tool for studying decision-making in face of classical uncertainty. The optimal betting strategy has been in the mathematical literature since the 1950s. Known as the classical Kelly criterion for optimal betting, it holds that under ``fair" odds and \textit{i.i.d.} repeated horse-races you should bet a fraction of capital on each horse that is proportional to the true winning probabilities of the latter in order to maximize the asymptotic growth rate of the cumulative wealth.\\
\indent In this paper, we presented a semi-classical model describing betting scenarios where the payoff of the gambler is encoded into the internal degrees of freedom of a quantum memory element. We provided a translation of the classical Kelly framework  into the new semi-classical playing field. Specifically, the invested capital was associated with the ergotropy of a single mode of the electromagnetic radiation;   the losses and winning events with the attenuation and amplification processes of the just-mentioned single-mode; instead, the (random) evolution of the capital, represented by the evolution of the quantum memory, is characterized within the theoretical setting of Bosonic Gaussian channels. As in the classical Kelly Criterion for optimal betting, we defined the  asymptotic doubling rate of the model and identified the optimal gambling strategy for fixed odds and winning probabilities. We carried out an accurate statistical and numerical analysis  to determine the performance of the model.
We found  that if the input capital state belongs to the set of Gaussian density matrices then the best option for the gambler is to devote all their initial resources into coherent state amplitude.\\
\indent A first natural problem to address is the generalization of our constructions to multimode scenarios. This would  pave the way for 
the  study of non-classical interference effects in betting games where the gambler can distribute their capital on parallel independent horse-races. In particular, this happens quite often in reality \cite{williams2005information, maclean2011kelly} and the comparison across different racetrack markets is fundamental in analysing their informational efficiency. Moreover, since the late nineties, platforms allowing short as well as long bets have been introduced (betting exchanges). This allows bettors to wager against over-priced horses just as hedge fund managers can short financial contracts.  Finding an adequate quantum setting for betting exchanges, including possibly general financial markets, would considerably extend the 
set of applications of our framework. 
Finally an interesting perspective is also offered by the possibility of employing the random
CPT trajectory model introduced in Sec.~\ref{sec:theProtocol} to analyse the lasing effects in random materials~\cite{Wiersma2008}.
\newline

The Authors would like to thank Giuseppe La Rocca for insightful discussions regarding the possible applications of the formalism to study  random lasing events. V. G. acknowledges support by MIUR via PRIN 2017 (Progetto di Ricerca di Interesse Nazionale): project QUSHIP (2017SRNBRK).
The authors acknowledge the financial support of UniCredit Bank R$\&$D group through the \emph{Dynamics and Information Research Institute} at the Scuola Normale Superiore.  \\

\bibliographystyle{unsrtnat}
\bibliography{kelly_article}

\newpage
\newpage

\begin{appendix}
\begin{lemma}[Recurrence relation for $\bar{\alpha}_t$]\label{lem:recursive}
The  recurrence relation~(\ref{DISTRALPHA}) 
admits the solution 
\begin{equation}\label{eq:recursivesolutionalpha}
\bar{\alpha}_t 
= \bar{g}^2_t \sum_{\ell = 1}^{t} \frac{\alpha_{j_\ell}}{\bar{g}_\ell^2}.
\end{equation}
\end{lemma}
\begin{proof} The parameters $\bar{\alpha}_t$ and $\bar{g}_{t}$ define the action of the BGC~(\ref{COMPO33}). By the equation \eqref{DISTRALPHA} we have: 
\begin{equation} \label{RECURRENCE_ALPHA}
\bar\alpha_{t+1} = g_{j_{t+1}}^2\bar\alpha_t + \alpha_{j_{t+1}}.
\end{equation} 
Now using $\bar\gamma_t$ as defined in \eqref{DEFBARGT}, we can find an explicit formula for $\bar\alpha_t$:
\begin{equation}
\bar\gamma_{t+1}-\bar\gamma_t = \frac{\bar{\alpha}_{t+1}}{\bar g_{t+1}^2}-\frac{\bar{\alpha}_{t}}{\bar g_{t}^2}=
\frac{\alpha_{j_{t+1}}}{\bar g_{t+1}^2}.
\nonumber
\end{equation}
\noindent So we can solve for $\bar\gamma_t$, integrating the finite difference:
\begin{equation}
\sum_{\ell=1}^{t-1} (\bar\gamma_{\ell+1} - \bar\gamma_{\ell}) = \bar\gamma_t - \bar\gamma_1 = \sum_{\ell=1}^{t-1} \frac{\alpha_{j_{\ell+1}}}{\overline{g}_{\ell+1}^2},
\end{equation}
since we have that $\bar\gamma_1 = \frac{\alpha_{j_{1}}}{\bar g^2_1}$ we obtain:
\begin{equation}
\bar\gamma_t = \sum_{\ell=1}^t \frac{\alpha_{j_{\ell}}}{\bar g_{\ell}^2},
\end{equation}
and recalling the definition \ref{eq:ell} of $\bar\gamma_t$ we obtain the thesis.
\end{proof}

We conclude this part by deriving some properties of $\bar\gamma_t$.

\begin{lemma}[Moments of $\bar\gamma_t$]\label{lem:momenta}
If we call 
$$\mathbb{E}\left[\frac{1}{g^2}\right] := \sum_{j=1}^J\frac{p_j}{g_j^2}\;, \qquad \mathbb{E}\left[\frac{1}{g^4}\right] := \sum_{j=1}^J\frac{p_j}{g_j^4}\;,$$
$$\mathbb{E}\left[\frac{\alpha}{g^2}\right] := \sum_{j=1}^Jp_j\frac{\alpha_j}{g_j^2} \;, \qquad \mathbb{E}\left[\frac{\alpha^2}{g^4}\right] := \sum_{j=1}^Jp_j\frac{\alpha_j^2}{g_j^4},$$
then
\begin{equation}\label{EXP1} 
\mathbb{E}\left[\bar\gamma_t\right] = \mathbb{E}\left[\frac{\alpha}{g^2}\right]\frac{1-\mathbb{E}\left[1/g^2\right]^{t}}{1- \mathbb{E}\left[1/g^2\right]}\;.
\end{equation}
and 
\begin{equation}
\begin{split}
&\mathbb{E}\left[\bar\gamma_t^2\right] =  \mathbb{E}\left[\frac{\alpha^2}{g^4}\right]\frac{1-\mathbb{E}\left[1/g^4\right]^t}{1-\mathbb{E}\left[1/g^4\right]} - \frac{1}{\mathbb{E}[1/g^2]-\mathbb{E}\left[1/g^4\right]}\\
&\times\left(\frac{1-\mathbb{E}\left[1/g^2\right]^t}{1-\mathbb{E}\left[1/g^2\right]}-\frac{1-\mathbb{E}\left[1/g^4\right]^t}{1-\mathbb{E}\left[1/g^4\right]}\right) \;.
\end{split}
\end{equation}
\end{lemma}

\begin{proof}
The expected value of $\bar\gamma_t$ is:
\begin{eqnarray}
\mathbb{E}[\bar\gamma_t] &=& \mathbb{E}\left[\frac{\bar{\alpha}_t}{\bar{g}^2_t}\right] = \sum_{\ell=1}^t \mathbb{E}\Bigg[\frac{\alpha}{g^2}\Bigg]\mathbb{E}\Bigg[\frac{1}{\bar g_{\ell-1}^2} \Bigg]\\
&=&\mathbb{E}\left[\frac{\alpha}{g^2}\right]\sum_{\ell = 1}^{t}\mathbb{E}\left[\frac{1}{g^2}\right]^{\ell-1} = \mathbb{E}\left[\frac{\alpha}{g^2}\right]\frac{1-\mathbb{E}\left[1/g^2\right]^{t}}{1- \mathbb{E}\left[1/g^2\right]},
\nonumber
\end{eqnarray}
where we used the fact that the different bets are independent.

Regarding the second moment of $\bar\gamma_t$, we split it into two terms as:
\begin{eqnarray}
\mathbb{E}\left[\bar\gamma_t^2\right] &=& \mathbb{E}\left[\left(\sum_{\ell = 1}^{t} \frac{\alpha_{j_\ell}}{g^2_{j_\ell}}\right)^2\right]\\
&=& \mathbb{E}\left[\sum_{\ell = 1}^{t}\sum_{m = 1}^t \frac{\alpha_{j_{\ell}}}{\bar{g}_\ell^2}\frac{\alpha_{j_m}}{\bar{g}^2_m}\right]\nonumber\\
&=&\mathbb{E}\left[\frac{\alpha^2}{g^4}\right]\frac{1-\mathbb{E}\left[1/g^4\right]^t}{1-\mathbb{E}\left[1/g^4\right]} + 2 \mathbb{E}\left[\sum_{m < \ell} \frac{\alpha_{j_\ell}}{\bar{g}^2_\ell}\frac{\alpha_{j_m}}{\bar{g}^2_m}\right]\nonumber\\
&:= & (I)+(II)\nonumber\;, 
\end{eqnarray}
\noindent where the second term can also be written as:
\begin{eqnarray}
(II) &=&2 \sum_{\ell = 2}^{t} \sum_{m = 1}^{\ell-1}\mathbb{E}\left[\alpha_{j_\ell} \alpha_{j_m} \left(\prod_{i = m + 1}^{\ell} g_{j_i}^2 \prod_{s = 1}^{m} g_{j_s}^4 \right)^{-1}\right]\nonumber\\
&=& 2 \mathbb{E}\left[\frac{\alpha}{g^4}\right] \mathbb{E}\left[\frac{\alpha}{g^2}\right]\sum_{\ell = 2}^{t}\sum_{m = 1}^{\ell-1}\mathbb{E}\left[\frac{1}{g^2}\right]^{\ell-m-1}\mathbb{E}\left[\frac{1}{g^4}\right]^{m-1}\nonumber\\
&=& 2 \mathbb{E}\left[\frac{\alpha_1}{g^4}\right] \mathbb{E}\left[\frac{\alpha_1}{g^2}\right]\sum_{\ell = 2}^{t}\mathbb{E}\left[\frac{1}{g^2}\right]^{\ell-2}\sum_{m = 1}^{\ell-1}\left(\frac{\mathbb{E}\left[1/g^4\right]}{\mathbb{E}\left[1/g^2\right]}\right)^{m-1}\nonumber\\
&=&\frac{1}{1-\mathbb{E}\left[1/g^4\right]/\mathbb{E}\left[1/g^2\right]}\nonumber\\
&&\times \sum_{\ell = 2}^{t}\mathbb{E}\left[\frac{1}{g^2}\right]^{\ell-2}\left(1-\left(\frac{\mathbb{E}\left[1/g^4\right]}{\mathbb{E}\left[1/g^2\right]}\right)^{\ell-1}\right)\nonumber\\
&=&\mathbb{E}\left[\frac{\alpha^2}{g^4}\right]\frac{1-\mathbb{E}\left[1/g^4\right]^t}{1-\mathbb{E}\left[1/g^4\right]} - \frac{1}{\mathbb{E}[1/g^2]-\mathbb{E}\left[1/g^4\right]}\nonumber\\
&&\times \left(\frac{1-\mathbb{E}\left[1/g^2\right]^t}{1-\mathbb{E}\left[1/g^2\right]}-\frac{1-\mathbb{E}\left[1/g^4\right]^t}{1-\mathbb{E}\left[1/g^4\right]}\right)\;.
\end{eqnarray}
\end{proof}

\begin{corollary} \label{cor:convmom}
Under the assumptions that $\mathbb{E}\left[\frac{1}{g^2}\right], \: \mathbb{E}\left[\frac{1}{g^4}\right] < 1$ and $\mathbb{E}\left[\frac{1}{g^2}\right]-\mathbb{E}\left[\frac{1}{g^4}\right] \neq 0$ it results that: \\
\begin{eqnarray} \label{COROLL}
&& \mathbb{E}\left[\lim_{t \rightarrow \infty}\bar\gamma_t\right] = \mathbb{E}\left[\frac{\alpha}{g^2}\right]\frac{1}{1- \mathbb{E}\left[1/g^2\right]} \;, \nonumber \\
&& \mathbb{E}\left[\lim_{t\rightarrow\infty}\bar\gamma_t^2\right] = \mathbb{E}\left[\frac{\alpha^2}{g^4}\right]\frac{1}{1-\mathbb{E}\left[1/g^4\right]} - \frac{1}{\mathbb{E}[1/g^2]-\mathbb{E}\left[1/g^4\right]} \nonumber \\
&&\times\left(\frac{1}{1-\mathbb{E}\left[1/g^2\right]}-\frac{1}{1-\mathbb{E}\left[1/g^4\right]}\right) \;. 
\end{eqnarray}
\end{corollary}

\begin{proof}
By definition $\forall$$c \in \mathbb{R}^+$ $\{\bar\gamma_t^{c}\}$ is a positive and monotonically increasing sequence, so for the monotone convergence theorem we have that:
$$\lim_{t\rightarrow\infty}\mathbb{E}\left[\bar\gamma_t^{c}\right] = \mathbb{E}\left[\lim_{t\rightarrow\infty}\bar\gamma_t^{c}\right].$$
If we take the limit of the formulas in lemma \eqref{lem:momenta} we obtain the thesis
\end{proof}

\begin{theorem} \label{convlt}
If $\bar{g}_t^2$ diverges exponentially almost surely  then $\bar{\gamma}_t$ converges almost surely. 
\end{theorem}
\begin{proof}
We just need to prove that the sequence $\{\bar{\gamma}_t\}_{t \in \mathbb{N}}$ is a Cauchy sequence almost everywhere in our probability space.
To formally embed the statement in a rigorous mathematical framework, we should introduce  a probability space $(\Omega, \Sigma, \mathcal{P})$ where   $\mathcal{P}$ the Bernoulli distribution on $\left\{1, \ldots, J\right\}$, 
$\Omega :=  \left\{1,\ldots,J\right\}^{\mathbb{N}}$ and 
$\Sigma$ to be the $\sigma$-algebra generated by the cylinder sets.
Then taking $\omega \in \Omega$ 
we identify $\bar{g}_{t}(\omega)$ with the 
define quantity (\ref{DEFBARGT} 
) computed on the first $t$ elements of $\omega$, and define $\bar G_t(\omega):=\frac{{\log_2} \bar{g}_{t}(\omega)
}{t}$. In this language Eq.~(\ref{eq:sllndd})  formally translates into  
\begin{equation}
\Pr[\omega \in \Omega : \forall \varepsilon > 0 \: \exists \, \bar{t} \; s.t. \; \forall \, t\geq \bar{t}\, \: |\bar{G}_t(\omega) - G| < \varepsilon] = 1,
\end{equation}
where $G$ is a short hand notation for $G(\boldsymbol{\eta}, \mathbf{k},  \mathbf{p})$.
Accordingly given 
\begin{equation}
\tilde{\Omega} := \{\omega \in \Omega : \forall \varepsilon > 0 \: \exists \, \bar{t} \; s.t. \; \forall \, t\geq \bar{t}\, \: |\bar{G}_t(\omega) - G| < \varepsilon\}\;, 
\label{eq:oneSet}
\end{equation}
it follows that 
\begin{eqnarray}
\Pr(\tilde{\Omega}):= \Pr(\omega \in \tilde{\Omega})=1\;. 
\end{eqnarray}  Now, let $\omega\in\tilde{\Omega}$. Fix $\varepsilon\in(0,\frac{G}{2})$, $\delta>0$ and let $\bar{t}(\omega,\varepsilon)$ as in  Eq.\eqref{eq:oneSet}. Then for all $t\geq\bar{t}(\omega,\varepsilon)$
\begin{equation}
2^{t(G-\varepsilon)} \leq \bar{g}_t(\omega) \leq 2^{t(G+\varepsilon)}.
\end{equation}
\noindent Let $\overline{m} > \bar{t}(\omega,\varepsilon)$, $p,q > \overline{m}$ with $p < q$. For a chosen $\omega$
\begin{eqnarray}
\bar{\gamma}_q(\omega) - \bar{\gamma}_p(\omega) & = & \sum_{\ell=p}^q\frac{\alpha_{j_\ell}(\omega)}{g_\ell^2(\omega)}\frac{1}{\bar{g}_{\ell-1}^2(\omega)} \nonumber \\
& \leq & K\sum_{\ell=p}^q 2^{-2(\ell-1)(G-\varepsilon)} \nonumber \\
& \leq & K'2^{-(\overline{m}-1)G},
\end{eqnarray}
\noindent for some constants $K,K' > 0$ -- we are using the convention of
indicating with $\bar{\gamma}_q(\omega)$ the quantity (\ref{eq:ell}) associated with the 
first $q$ elements of $\omega$.  For a sufficiently large $\overline{m}$ we have $K'2^{-(\overline{m}-1)G} < \delta \; \forall \delta > 0$. Hence we proved that 
\begin{eqnarray}
\Pr(\omega \in \Omega : && \forall \delta > 0 \: \exists \, \overline{m} \; s.t. \; \\ \nonumber 
&&\forall \, p,q\geq \overline{m} \;, |\bar{
\gamma}_p(\omega)-\bar{\gamma}_q(\omega)| < \delta)=1,
\end{eqnarray}
\noindent or equivalently that $\{\bar{\gamma}_t\}_{t \in \mathbb{N}}$ is a Cauchy sequence almost everywhere.
\end{proof}

\end{appendix}

\end{document}